\documentclass[leqno]{amsart}
\usepackage{cases}
\usepackage{amsfonts,amssymb,amsmath,amsgen,amsthm}
\usepackage{hyperref}

\theoremstyle{plain}
\newtheorem{theorem}{Theorem}[section]
\newtheorem{definition}[theorem]{Definition}

\newtheorem{lemma}[theorem]{Lemma}

\newtheorem{proposition}[theorem]{Proposition}
\newtheorem{hyp}[theorem]{Assumption}
\theoremstyle{remark}
\newtheorem{remark}[theorem]{Remark}
\newtheorem*{notation}{Notation}

\def\R{{\mathbf R}}
\def\N{{\mathbf N}}
\def\O{\mathcal O}

\def\({\left(}
\def\){\right)}
\def\<{\left\langle}
\def\>{\right\rangle}
\def\le{\leqslant}
\def\ge{\geqslant}

\def\Tend#1#2{\mathop{\longrightarrow}\limits_{#1\rightarrow#2}}

\def\d{{\partial}}
\def\eps{\varepsilon}
\def\l{\lambda}

\def\si{{\sigma}}
\def\g{\gamma}

\DeclareMathOperator{\RE}{Re}
 \DeclareMathOperator{\IM}{Im}
 \numberwithin{equation}{section}

\begin{document}

\title[Wave packets for Hartree
equations]{Semiclassical wave packet dynamics for Hartree equations}  
  \author[P. Cao]{Pei Cao}
\address{Department of Mathematical Science\\ Tsinghua University
\\Beijing 100084\\ China}
\email{pcao04@gmail.com}
\author[R. Carles]{R\'emi Carles}
\address{Univ. Montpellier~2\\Math\'ematiques
\\CC~051\\F-34095 Montpellier}
\address{CNRS, UMR 5149\\  F-34095 Montpellier\\ France}
\email{Remi.Carles@math.cnrs.fr}

\begin{abstract}
  We study the propagation of wave packets for
  nonlinear nonlocal Schr\"odinger equations 
   in the semi-classical limit. When the kernel is smooth, we
   construct approximate solutions for the wave functions in
   subcritical,  critical
   and supercritical cases (in terms of the size of the initial data).
   The validity of the approximation is proved up to
    Ehrenfest time.
 For homogeneous kernels, we establish similar results in subcritical
 and critical cases.
Nonlinear superposition principle for two nonlinear wave packets is
also considered. 
\end{abstract}
\thanks{This work was supported by the French ANR project
  R.A.S. (ANR-08-JCJC-0124-01), and was achieved when the first author
was visiting the University of Montpellier~2, under a grant from
Tsinghua University. She would like to thank these institutions for
this opportunity.}
\maketitle

\section{Introduction}
In this paper, we consider the
following semi-classically scaled Hartree equation
\begin{equation}
  \label{eq:hartreegen}
  i\eps \d_t \psi^\eps + \frac{\eps^2}{2}\Delta \psi^\eps =
  V(t,x)\psi^\eps +\(K\ast 
  |\psi^\eps|^2\) \psi^\eps,\quad (t,x)\in \R_{+}\times \R^d,
\end{equation}
where $K:\R^d\to \R$, $V:\R_{+}\times \R^d\to \R$, $d\ge 1$, with
initial data
\begin{equation}
  \label{eq:ci}
  \psi^\eps(0,x)=\eps^M\times \eps^{-d/4}
  a\left(\frac{x-x_0}{\sqrt\eps}\right)
e^{i{(x-x_0)\cdot\xi_0/\eps}},\quad a\in{\mathcal S}({\R}^d),\quad
x_0,\xi_0\in \R^d.
\end{equation}
Such data, which are called \emph{semi-classical wave packets} (or
\emph{coherent states}), have raised great interest in the linear
case (see e.g. \cite{BR01,CR97,CR06,L86,P97}). It is well known that
if the data is a wave packet, then the solution of the linear
equation ($K=0$) associated with \eqref{eq:hartreegen} and
\eqref{eq:ci} still is a wave packet at leading order up to times of
order $C\log \left(\frac{1}{\eps}\right)$,  called \emph{Ehrenfest
  time} (see e.g. 
\cite{BGP99,HJ00,HJ01}). We refer the reader to the recent papers
\cite{CR97,Rob10,Rou-p,SR09}, where overview and  references on the
topics can be found. Throughout this paper, we consider dynamical
properties for positive time only: this is just for the simplicity of
notations, since the equation is reversible.
\smallbreak

 This paper is inspired by the two recent papers \cite{APPP-p}, where
 \eqref{eq:hartreegen} is considered for a smooth kernel $K$, and
 \cite{CaFe11} where the nonlinearity is local, as opposed to the
 Hartree nonlinearity. In \cite{CaFe11}, the authors proved that if
 the initial data 
have subcritical size, the leading order behavior of the wave
function as $\eps\rightarrow 0$ is the same as for the linear equation.
When the size of the initial data is critical, at leading order the wave
function propagates like a coherent state whose envelope is given by
a nonlinear equation, up to a nonlinear analogue of the Ehrenfest
time. In this paper, we follow a similar approach in the case of 
nonlocal Schr\"odinger equations, a case where this analysis is not
\emph{a priori} clear, precisely because the nonlinearity is nonlocal.

Up to changing $\psi^{\varepsilon}$ to
$\varepsilon^{-M}\psi^{\varepsilon},$  \eqref{eq:hartreegen} and
\eqref{eq:ci} can be written as:
\begin{equation}\label{eq:NLS0}
 \left\{
\begin{aligned}
     i\varepsilon\partial_{t} \psi^{\varepsilon}+\frac{\varepsilon^{2}}{2} \Delta
\psi^{\varepsilon}&=V\(t,  x\)
\psi^{\varepsilon}+\varepsilon^{\alpha} \(K*|\psi^{\varepsilon}|^{2}
\)\psi^{\varepsilon},\\
      \psi^{\varepsilon}(0,x)&= \varepsilon^{-d/4}  a
      \(\frac{x-x_{0}} {\sqrt{\varepsilon}}\) e^{i(x-x_{0})\cdot
        \xi_{0}/ \varepsilon}, 
  \end{aligned}
\right.
\end{equation}
with $\alpha=2M$. Notice that the initial data are of order
$\O(1)$ in $L^{2}(\R^{d})$, and $\alpha$ accounts for the strength of
nonlinear effects in the limit $\eps\to 0$.

Consider the trajectories associated with the Hamiltonian flow
$\frac{|\xi|^{2}}{2}+V\(t, x\(t\)\)$:
\begin{equation}\label{eq:traj}
 \dot x(t)=\xi(t),\;\;\dot \xi(t)=-\nabla V\(t,x(t)\);\quad
 x(0)=x_0,\;\xi(0)=\xi_0.
 \end{equation}

 \begin{hyp}\label{hyp:V}
  The external potential $V$ is smooth, real-valued, and at most quadratic
  in space:
  \begin{equation*}
 V\in C^\infty(\R_{+}\times \R^d;\R),\quad \text{and}\quad   \d_x^\beta V\in
 L^\infty\(\R_{+}\times\R^d\),\quad \forall |\beta|\ge  2.
  \end{equation*}
In addition, we require $t\mapsto \nabla V(t,0)\in
 L^\infty(\R_{+})$.
\end{hyp}
\begin{remark}
  If $V=V(x)$ does not depend on time, the last assumption is
  automatically fulfilled. This assumption is needed to ensure that
  the Hamiltonian flow grows at most exponentially in time. Typically,
  if $V(x)=\kappa \cdot x e^{e^t}$ for some (constant) $\kappa\in \R^d$, then
$\dot \xi(t) = -\kappa e^{e^t}$, so $x$ and $\xi$ grow like a double
exponential.
\end{remark}
The following lemma is straightforward.
\begin{lemma}\label{lem:traj}
Let $(x_0,\xi_0)\in \R^d\times \R^d$. Under Assumption~\ref{hyp:V},
\eqref{eq:traj} has a unique global, smooth solution $(x,\xi)\in
C^\infty(\R_{+};\R^d)^2$. It grows at most exponentially:
 \begin{equation}\label{growthtraj}
 \exists C_0>0,\quad \left|x(t)\right|+\left|\xi(t)\right|\lesssim
 e^{C_0t},\quad \forall t\ge  0.
 \end{equation}
 \end{lemma}

As far as the Hartree kernel is concerned, two cases will be considered,
leading to two different interesting  phenomena:
\begin{itemize}
\item Smooth kernel: $K\in W^{3,\infty}(\R^d)$, with $K$ smooth near the origin.
\item Homogeneous kernel: $K(x)=\l |x|^{-\gamma}$, with $\l\in \R$
  and $0<\gamma<\min (2,d)$. 
\end{itemize}
The second case includes the three-dimensional Schr\"odinger--Poisson,
typically. 
\begin{remark}
  For several results (linearizable case --- see definition below ---
  or finite time propagation), the second assumption
  could be relaxed to $0<\gamma<\min(4,d)$ (energy-subcritical
  case). To simplify the 
  presentation, we shall not discuss this extension.
\end{remark}

We will focus on the first case in Section ~\ref{sec:smooth}, in
which we mostly revisit the results from \cite{APPP-p}. 
In the rest of this introduction, we consider the homogeneous case. We
seek the solution with the form
\begin{equation*}
\psi^{\varepsilon}(t,x)=\varepsilon^{-d/4}u^{\varepsilon}\(t
,\frac{x-x(t)}{\sqrt{\varepsilon}}\)e^{i(S(t)+
  \xi(t)\cdot(x-x(t)))/\varepsilon} .
\end{equation*}
Here $S(t)$ is the classical Lagrangian action along the Hamiltonian
flow generated by ~\eqref{eq:traj}, given by
 \begin{equation}\label{eq:classicalaction}
S(t)=\int_0^t  \(\frac{1}{2} |\xi(s)|^2-V(s,x(s))\)ds.
\end{equation}
In terms of $u^{\eps}=u^{\eps}(t, y)$, ~\eqref{eq:NLS0} is
equivalent
\begin{equation}\label{eq:hartree}
 i\partial_{t}u^{\varepsilon}+\frac{1}{2}\Delta u^{\varepsilon}=
 V^{\varepsilon} (t, y) u^{\varepsilon}+\lambda 
     \varepsilon^{\alpha-\alpha_{c}}\(|y|^{-\gamma}\ast
     |u^{\varepsilon}|^{2}\)u^{\varepsilon}
\end{equation}
with the initial date $u^{\eps}(0,y)=a(y)$, where
\begin{equation*}
\alpha_{c}=1+\frac{\gamma}{2}
\end{equation*}
is a critical exponent and the time-dependent potential
$V^{\varepsilon}(t, y)$ is given by
\begin{eqnarray*}
V^{\varepsilon}\(t, y\)=\frac{1}{\varepsilon}\(V(t,
x(t)+\sqrt{\varepsilon}y)-V\(t, x(t)\)-\sqrt{\varepsilon}\<\nabla
V(t, x(t)), y\>\).
\end{eqnarray*}
It reveals the first terms of the Taylor expansion of $V$ about the
point $x(t).$ Passing formally to the limit, $V^{\varepsilon}$
converges to the Hessian of $V$ at $x(t)$ evaluated at $(y,y)$. 
Throughout the paper, we denote
\begin{equation}\label{eq:Q}
Q(t)={\rm Hess}\ V\(t, x(t)\).
\end{equation}

\subsection{The linear case $\lambda=0$}\label{sec:linear} 
Introduce the function
\begin{equation*}\label{eq:linearvarphi}
\varphi_{{\rm
lin}}^{\varepsilon}(t,x)=\varepsilon^{-d/4}u_{\rm lin}\(t,
\frac{x-x(t)}{\sqrt{\varepsilon}}\)e^{i(S(t) 
+\xi(t)\cdot(x-x(t)))/\varepsilon},
\end{equation*}
where $u_{\rm lin}$ solves
\begin{equation}\label{eq:linearu}
 i\partial_{t}u_{\rm lin}+\frac{1}{2}\Delta u_{\rm
   lin}=\frac{1}{2}\<y, Q(t)y\>u_{\rm lin}\quad ;\quad 
u_{\rm lin}(0,y)=a(y).
\end{equation}
Then the following lemma is well-known, see e.g.
\cite{BGP99,CaFe11,CR97,CR06,CR07,H80,HJ00,HJ01} and references
therein.
\begin{lemma}\label{lemlinear}
Let $a\in \mathcal{S}(\R^{d})$, and $\psi^\eps$ solve \eqref{eq:NLS0}
with $K=0$. There exist positive constants $C$
and $C_{1}$ independent of $\varepsilon$
 such that
 \begin{equation*}
 \|\psi^{\varepsilon}(t)-\varphi_{{\rm
     lin}}^{\varepsilon}(t)\|_{L^{2}(\R^{d})}\le 
 C\sqrt{\varepsilon} e^{C_{1}t}.
 \end{equation*}
 In particular, there exists  $c>0$ independent of $\eps$ such that
 \begin{equation*}
 \sup_{0\le t\le c\ln\frac{1}{\eps}}\|\psi^{\varepsilon}(t)
 -\varphi_{{\rm lin}}^{\varepsilon}(t)\|_{L^{2}(\R^{d})}\Tend
 \eps 0 0. 
 \end{equation*}
 \end{lemma}
\subsection{The nonlinear case $\lambda\neq 0$}
As in \cite{CaFe11}, we introduce two linear operators, which are
essentially $\nabla$ and $x$, up to the wave 
packet scaling, in the moving frame:
\begin{equation*}
  A^\eps(t) =  \sqrt \eps \nabla -i\frac{\xi(t)}{\sqrt
    \eps}\quad ;\quad
  B^\eps(t)=\frac{x-x(t)}{\sqrt\eps}.
\end{equation*}
For $f\in \Sigma:=\{f\in H^{1}(\R^{d}); xf\in L^{2}(\R^{d})\}$,
we define
\begin{equation*}
\|f\|_{\mathcal{H}}=\|f\|_{L^{2}(\R^{d})}+\|A^{\varepsilon}
f\|_{L^{2}(\R^{d})} +\|B^{\varepsilon}f\|_{L^{2}(\R^{d})}.
\end{equation*}
\subsubsection{The subcritical case $\alpha>\alpha_{c}$}
In this case, the solution of \eqref{eq:NLS0} is linearizable in
the sense of \cite{PG96}: $\varphi_{{\rm lin}}$
yields a good approximation to $\psi^{\eps}$, up to Ehrenfest time.
\begin{proposition}\label{prop:subhartree}
 Let $\lambda\in \R$, $0<\gamma<\min\(2, d\)$ and $\alpha>\alpha_c$. Suppose
that $a\in \mathcal{S}(\R^{d})$ and $V$ satisfies Assumption
~\ref{hyp:V}. Then there exist positive constants $C, C_{1}, C_{2}$
independent of $\varepsilon$, and $\varepsilon_{0}>0$ such that for any
$\varepsilon\in ]0, \varepsilon_{0}]$,
\begin{equation*}
\|\psi^{\varepsilon}(t)-\varphi_{{\rm
    lin}}^{\varepsilon}(t)\|_{\mathcal{H}}\le
C\varepsilon^{\kappa}e^{C_{1}t}, 
\quad 0\le t\le C_{2}\ln\frac{1}{\varepsilon},\quad
\kappa=\min\(\frac{1}{2}, \alpha-\alpha_{c}\).
\end{equation*}
In particular, there exists a positive constant $c$ independent of
$\varepsilon$ such that
\begin{equation*}
\sup_{0\le t\le c\ln\frac{1}{\varepsilon}}\|\psi^{\varepsilon}(t)-
\varphi_{\rm{lin}}^{\varepsilon}(t)\|_{\mathcal{H}}\Tend \eps 0 0. 
\end{equation*}
\end{proposition}
\subsubsection{The critical case $\alpha=\alpha_{c}$}
By passing formally to the limit $\varepsilon\rightarrow 0,$
  ~\eqref{eq:hartree} can be written as
  \begin{equation}\label{eq:nlhartree}
   i\partial_{t}u+\frac{1}{2}\Delta u=\frac{1}{2} \<y, Q(t)y\> u+\lambda
\(|y|^{-\gamma}\ast |u|^{2}\) u\quad ;\quad 
u(0,y)=a(y).
\end{equation}
The Cauchy problem for \eqref{eq:nlhartree} is addressed in
\S\ref{sec:profile}. 
For $\alpha=\alpha_c$, the solution to \eqref{eq:NLS0} is not
linearizable: the nonlinearity affects the dynamics at leading
order. For $k\in \N$, define ($\Sigma^1=\Sigma$)
\begin{equation*}
  \Sigma^k = \left\{ f\in L^2(\R^d)\ ;\ \| f \| _{\Sigma^k}:=
    \sum_{|\alpha|+|\beta|\le  k}\left\lVert x^\alpha \d_x^\beta
      f\right\rVert_{L^2(\R^d)}<\infty\right\}. 
\end{equation*}
We prove:
\begin{theorem}\label{thm:criticalhartree}
Let $\lambda\in\R$, $0<\gamma<\min\(2, d\)$, $\alpha=\alpha_c$ and $a\in
\Sigma^3$. Suppose that $V$ satisfies Assumption~\ref{hyp:V}. Let $u\in
C(\R_+; \Sigma^3)$ be the solution to \eqref{eq:nlhartree} and
\begin{equation}\label{eq:varphi}
\varphi^{\varepsilon}(t,x)=\varepsilon^{-d/4}u\(t,
\frac{x-x(t)}{\sqrt{\varepsilon}}\)e^{i(S(t)+\xi(t)\cdot
(x-x(t)))/\varepsilon}.
\end{equation}
Then there exist
positive constants $C, C_{1}, C_{2}$ independent of $\varepsilon$,
and $\varepsilon_{0}>0$ such that for any $\varepsilon\in ]0,
\varepsilon_{0}]$,
\begin{equation*}
\|\psi^{\varepsilon}(t)-\varphi^{\varepsilon}(t)\|_{L^{2}(\R^{d})}\le
C\sqrt{\varepsilon}\exp(C_{1}t),\quad 
0\le t\le C_{2}\ln\frac{1}{\varepsilon}.
\end{equation*}
In particular, there exists a positive constant $c$ independent of
$\eps$ such that
\begin{equation*}
\sup_{0\le t\le c\ln\frac{1}{\varepsilon}}\|\psi^{\varepsilon}(t)-
\varphi^{\varepsilon}(t )\|_{L^{2}(\R^{d})}\Tend \eps 0 0.
\end{equation*}
Furthermore, if $a\in\Sigma^4$, then for the same
constants $C_1,C_2$ as above,
\begin{equation*}
\|\psi^{\varepsilon}(t)-\varphi^{\varepsilon}(t)\|_{\mathcal{H}}\le
C_3\sqrt{\varepsilon}\exp(C_{1}t), \quad 
0\le t\le C_{2}\ln\frac{1}{\varepsilon}.
\end{equation*}
\end{theorem}

\begin{remark}[Notion of criticality] We will see
  in Section~\ref{sec:smooth} that when $K$ is smooth near the origin,
  the critical value of $\alpha$ is $\alpha_c=1$, in sharp contrast
  with the homogeneous case. For more general Hartree kernels, the
  picture should be as follows.
Assume that there exists $\l\in \R\setminus\{0\}$, $\g\ge 0$,
$\delta>0$ such that 
  \begin{equation*}
    K(x)=\l |x|^{-\g} + \O\(|x|^{-\g +\delta}\) \text{ as }x\to 0,
  \end{equation*}
and that $K$ is smooth, bounded as well as its derivatives, away from
the origin. Then we expect $\alpha_c =1 +\g/2$, with critical
phenomena similar to the cases studied in this paper: like in the smooth
kernel case if
$\g=0$, and like in the homogeneous kernel case if $\g>0$ (since wave
packets are  
extremely localized, the behavior of $K$ near the origin should be the
only relevant one). 
\end{remark}
\subsection{Nonlinear superposition} In this paragraph, we consider
the critical case $\alpha=\alpha_{c}$. Suppose that initial data
have the form 
\begin{equation*}
\psi^{\varepsilon}(0,x)=\eps^{-d/4}a_{1}\(\frac{x-x_{1}}{\sqrt{\varepsilon}}\)
e^{i(x-x_{1})\cdot\xi_{1}/\varepsilon}+
\eps^{-d/4}a_{2}\(\frac{x-x_{2}}{\sqrt{\varepsilon}}\) 
e^{i(x-x_{2})\cdot\xi_{2}/\varepsilon},
\end{equation*}
where $a_{1}$, $a_{2}\in \mathcal{S}(\R^{d})$, and
$(x_{1}, \xi_{1})\neq (x_{2}, \xi_{2})$. For $j\in \{1, 2\}$,
$(x_{j}(t), \xi_{j}(t))$ are the trajectories solutions to
\eqref{eq:traj} with initial data $(x_{j}, \xi_{j})$. Let
$S_{j}(t)$ be the classical action associated with $(x_{j}(t),
\xi_{j}(t))$ given by \eqref{eq:classicalaction} and $u_{j}$ be the
solutions of \eqref{eq:nlhartree} with initial data $a_{j}$.
Assume the $\varphi^{\varepsilon}_{j}$'s are defined as in
\eqref{eq:varphi}, and $\psi^{\varepsilon}\in C(\R_{+}; \Sigma)$ is
the solution to \eqref{eq:NLS0}.
As in \cite{CaFe11}, for $f\in \Sigma$, define
\begin{equation*}
\|f\|_{\Sigma_{\varepsilon}}=\|f\|_{L^{2}(\R^{d})}+\|\varepsilon
\nabla f\|_{L^{2}(\R^{d})}+\|xf\|_{L^{2}(\R^{d})}.
\end{equation*}
 For
bounded time, we have
\begin{theorem}\label{thm:superpositionbounded}
Let $0<\gamma<\min (2, d)$  and $a_{1}, a_{2}\in
\mathcal{S}(\R^{d})$. For any $T>0$ independent of $\varepsilon$,
\begin{equation*}
\sup_{0\le t\le T}\left\|\psi^{\varepsilon}(t)-
\varphi^{\varepsilon}_{1}(t)-\varphi^{\varepsilon}_{2}(t)\right\|_{\Sigma_{\varepsilon}}
=\mathcal{O}\(\varepsilon^{\frac{\gamma}{2(1+\gamma)}}\).
\end{equation*}
\end{theorem}
When time becomes large, we can establish a superposition property
for $d=1$, like in \cite{CaFe11}, where the condition $\g<\min (2,d)$
boils down to $\g<1$. 
\begin{theorem}\label{thm:superpositionlargetime}
Let $d=1$, $0<\gamma<1$ and $a_{1}, a_{2}\in
\mathcal{S}(\R^{d})$. Assume that $V$ does not depend on time, and
define
\begin{equation*}
E_{j}=\frac{\xi_{j}^{2}}{2}+V(x_{j}).
\end{equation*}
Suppose $E_{1}\neq E_{2}$. There exist positive constants $C, C_{1}, C_{2}$
independent of $\varepsilon$, and $\varepsilon_{0}>0$ such that for all
$\varepsilon\in ]0, \varepsilon_{0}]$,
\begin{equation*}
\|\psi^{\varepsilon}(t)-\varphi^{\varepsilon}_{1}(t)-
\varphi^{\varepsilon}_{2}(t)\|_{\Sigma_{\varepsilon}}\le
C\varepsilon^{\frac{\gamma}{2(1+\gamma)}}e^{C_{1}t}, \quad 
0\le t\le C_{2}\ln\frac{1}{\varepsilon}.
\end{equation*}
In particular, there exists a positive constant $c$ independent of
$\varepsilon$ such that
\begin{equation*}
\sup_{0\le t\le c\ln\frac{1}{\varepsilon}}\|\psi^{\varepsilon}(t)
-\varphi^{\varepsilon}_{1}(t)
-\varphi^{\varepsilon}_{2}(t)\|_{\Sigma_{\varepsilon}}\Tend \eps 0 0 .
\end{equation*}
\end{theorem}
\begin{notation}
Throughout the paper, $C$ denotes a constant independent of $\eps$
and $t$, whose value may change from one 
line to the other. For two positive numbers $a^\eps$ and $b^\eps$,
the notation $ a^\eps\lesssim b^\eps$ means that there exists $C>0$
\emph{independent of} $\eps$ such that for all $\eps\in ]0,1]$,
$a^\eps\le Cb^\eps$.
\end{notation}

\section{Smooth kernel} \label{sec:smooth}
  In this section, we shall assume that the kernel $K$ satisfies:
 \begin{hyp}\label{hyp:smooth}
  The kernel is bounded as well as  its first three derivatives, and smooth
  near the origin: for some neighborhood $\omega$ of the origin in $\R^d$,
  \begin{equation*}
  K\in W^{3,\infty}(\R^d)\cap C^3(\omega).
  \end{equation*}
\end{hyp}
This assumption is similar to the one made in \cite{APPP-p}. We
first establish the local and global existence for the solution for
\eqref{eq:NLS0} at the $L^{2}$ level.

\subsection{Construction of the exact solution}
We prove that for fixed $\eps>0$, \eqref{eq:hartreegen} has a unique,
global in time solution under Assumption~\ref{hyp:V}. Since for such a
result, $\eps$ is irrelevant, we shall consider the case $\eps=1$. 

 \begin{lemma}\label{lem:smoothglobal}
  Let $V$ satisfy Assumption~\ref{hyp:V}, $K\in L^\infty(\R^d)$ and $\psi_0\in
  L^2(\R^d)$. There exists 
  a unique solution  $\psi \in C\(\R_+;L^2(\R^d)\)$ to
  \begin{equation*}
    i\d_t \psi +\frac{1}{2}\Delta \psi = V(t,x)\psi + \(K\ast
    |\psi|^2\)\psi\quad ;\quad \psi_{\mid t=0}=\psi_0.
  \end{equation*}
 In addition, it satisfies
  $  \|\psi(t)\|_{L^2(\R^d)} = \|\psi_0\|_{L^2(\R^d)}$ for all time
  $t\ge 0$.
\end{lemma}
 \begin{proof}
Since $V$ may depend on time, we consider the more general Cauchy
problem with a varying initial time:
 \begin{equation}\label{eq:NLS1}
i\d_t \psi +\frac{1}{2}\Delta \psi = V(t,x)\psi + \(K\ast
    |\psi|^2\)\psi\quad ;\quad \psi_{\mid t=s}=\psi_s,
 \end{equation}
with $\psi_s\in L^2(\R^d)$. In view of Assumption~\ref{hyp:V}, the
linear case generates a unitary semigroup (\cite{Fujiwara79,Fujiwara}),
which we denote by $U(t,s)$: $\psi(t)= U(t,s)\psi_s$ when
$K=0$.
  
In the nonlinear case, using  Duhamel's formula, we can write
~\eqref{eq:NLS1} as 
\begin{equation*}
\psi(t)=U(t,s)\psi_s-i\int_{s}^{t}
U(t,\tau)\(K\ast |\psi|^{2}\)\psi(\tau)d\tau:=\Phi^{s}(\psi)(t).
\end{equation*}
For $s\ge 0$ and $T>0$, denote $I_{s,T}=[s, s+T]$ and introduce the space
\begin{equation*}
X_{s,T}=\left\{\psi\in C\(I_{s,T}, L^{2}(\R^{d})\);
\|\psi\|_{L^{\infty}(I_{s,T}; L^{2}(\R^{d}))}\le
2\|\psi_s\|_{L^{2}(\R^{d})}\right\}.
\end{equation*}
Let $\psi,
\psi_{1},\psi_{2}\in X_{s,T}$. Using  H\"older
inequality, we have
\begin{align*}
\|\Phi^{s}(\psi)\|_{L^\infty(I_{s,T}; L^{2}(\R^{d}))}&\le
\|\psi_s\|_{L^{2}(\R^{d})}+\|K*|\psi|^{2}\psi\|_{L^{1}(I_{s,T};
L^{2}(\R^{d}))}\\
&\le
\|\psi_s\|_{L^{2}(\R^{d})}+T\|K\|_{L^\infty(\R^d)}\|\psi\|_{L^{\infty}(I_{T};
L^{2}(\R^{d}))}^{3}\\
&\le 
\|\psi_s\|_{L^{2}(\R^{d})}+8T\|K\|_{L^\infty(\R^d)}\|\psi_s\|_{L^{2}(\R^{d})}^{3}.
\end{align*}
Observe that
\begin{align*}
&K\ast |\psi_{1}|^{2}\psi_{1}-K\ast |\psi_{2}|^2\psi_{2}\\
=&\frac{1}{2}K*(|\psi_{1}|^{2}+|\psi_{2}|^{2})(\psi_{1}-\psi_{2})+\frac{1}{2}K*
(|\psi_{1}|^{2}-|\psi_{2}|^{2})(\psi_{1}+\psi_{2}).
\end{align*}
Then by similar arguments as above, we have
\begin{equation*}
\|\Phi^{s}(\psi_{1})-\Phi^{s}(\psi_{2})\|_{L^\infty(I_{s,T};
 L^{2})}\le C
T\|K\|_{L^\infty}\|\psi_s\|_{L^{2}}^2
\|\psi_1-\psi_2\|_{L^\infty(I_{s,T};  L^{2})}.
\end{equation*}
Taking $T$ small enough,
we conclude $\Phi^{s}$ is a contraction from $X_{s,T}$ into itself, and
there exists a unique local solution $\psi\in C(I_{s,T};
L^{2}(\R^{d}))$ to \eqref{eq:NLS1}. By classical arguments, the
$L^2$-norm of $\psi$ does not depend on time, and since $T$ depends
only on $\|\psi_s\|_{L^2}$, the solution is global in time. 
 \end{proof}

 \subsection{The general strategy}

 As in \cite{CaFe11}, seek an approximate solution of the form
\begin{equation}
  \label{eq:approx}
  \varphi^\eps(t,x)=\eps^{-d/4} u 
\left(t,\frac{x-x(t)}{\sqrt\eps}\right)e^{i\left(S(t)+\xi(t)\cdot
    (x-x(t))\right)/\eps},
\end{equation}
for some profile $u$ independent of $\eps$, and some function $S(t)$
to be determined. When $K=0$, $S$ is the 
\emph{classical action} defined in \eqref{eq:classicalaction}.
We will see that according to the value $\alpha$ in \eqref{eq:NLS0},
the expression of $S$ may vary, accounting for nonlinear effects due to the
presence of the Hartree nonlinearity. 

In the cases $\alpha=0,1/2,1$ and $\alpha>1$, we will see that we can write
\begin{equation}\label{eq:generalstrategy}
  \begin{aligned}
   i\eps \d_t \varphi^\eps + \frac{\eps^2}{2}\Delta \varphi^\eps &-
   V\varphi^\eps -\eps^\alpha\(K\ast 
  |\varphi^\eps|^2\)\varphi^\eps  = \\
&\eps^{-d/4}e^{i\left(S(t)+\xi(t)\cdot
    (x-x(t))\right)/\eps}\(b_0+\sqrt \eps b_1 + \eps b_2 + \eps r^\eps\),
  \end{aligned}
\end{equation}
for $b_0, b_1,b_2$ independent of $\eps$. The approximate solution
$\varphi^\eps$ will be determined by the conditions
\begin{equation*}
  b_0=b_1=b_2=0.
\end{equation*}
The remaining factor $r^\eps$ will account for the error between the
exact solution $\psi^\eps$ and the approximate solution
$\varphi^\eps$. Denote 
\begin{equation*}
  \phi(t,x) = S(t)+\xi(t)\cdot    (x-x(t)).
\end{equation*}
The linear terms are computed as follows:
\begin{align*}
  i\eps\d_t \varphi^\eps &= \eps^{-d/4}e^{i\phi(t,x)/\eps}\(i\eps \d_t
  u - i\sqrt \eps \dot x(t)\cdot \nabla u - u\d_t \phi \).\\
\frac{\eps^2}{2}\Delta \varphi^\eps &=
\eps^{-d/4}e^{i\phi(t,x)/\eps}\( \frac{\eps}{2}\Delta u +i\sqrt\eps
\xi(t)\cdot \nabla u -\frac{|\xi(t)|^2}{2}u\). 
\end{align*}
Here, as well as below, one should remember that the functions are
assessed as in \eqref{eq:approx}. Recalling that the relevant space
variable for $u$ is  
\begin{equation*}
  y=\frac{x-x(t)}{\sqrt\eps},
\end{equation*}
we have:
\begin{equation*}
  \d_t \phi = \dot S(t) + \frac{d}{dt}\(\xi(t)\cdot (x-x(t))\) =
  \dot S(t) +\sqrt \eps \dot \xi(t)\cdot y-\xi(t)\cdot \dot x(t).
\end{equation*}
For the linear potential term, we compute, in terms of the variable $y$,
\begin{equation*}
  V\varphi^\eps =
  V(t,x)\eps^{-d/4}e^{i\phi(t,x)/\eps}u\(t,y\)=
\eps^{-d/4}e^{i\phi(t,x)/\eps} V\(t,x(t)+y\sqrt\eps\) u\(t,y\),
\end{equation*}
and we perform a Taylor expansion for $V$ about $(t,x(t))$:
\begin{align*}
  V\(t,x(t)+y\sqrt\eps\)u(t,y) &= V\(t,x(t)\) u(t,y)+ \sqrt \eps 
y\cdot \nabla V\(t,x(t)\)u(t,y)\\
&\quad  + \frac{\eps}{2}\<y,\nabla^2
V\(t,x(t)\)y\>u(t,y) 
+ \eps^{3/2} r_V^\eps(t,y),
\end{align*}
with
\begin{equation*}
  |r_V^\eps(t,y)|\le C\<y\>^3|u(t,y)|,
\end{equation*}
for some $C$ independent of $\eps$, $t$ and $y$, in view of
Assumption~\ref{hyp:V}. 
In the case $K=0$, we come up with the relations:
\begin{align*}
  b_0^{\rm lin} &= -u\( \dot S(t)-\xi(t)\cdot \dot x(t)
  +\frac{|\xi(t)|^2}{2}+V\(t,x(t)\)\) .\\
b_1^{\rm lin}&= -i\( \dot x(t)-\xi(t)\)\cdot \nabla u-y\cdot \(\dot \xi(t)+
\nabla V\(t,x(t)\)\) u.\\
b_2^{\rm lin} &=  i\d_t u +\frac{1}{2}\Delta u -
\frac{1}{2}\<y,\nabla^2 V\(t,x(t)\)y\>u. 
\end{align*}
For the nonlinear term, we compute similarly
\begin{align*}
  \(K\ast
  |\varphi^\eps|^2\)(t,x)&= \int K(x-z)|\varphi^\eps(t,z)|^2dz\\
  &=\eps^{-d/2} \int K(x-z)\left\lvert
    u\(t,\frac{z-x(t)}{\sqrt\eps}\)\right\rvert^2dz \\
& = \int K\(x-x(t)-z\sqrt\eps\)\left\lvert
    u\(t,z\)\right\rvert^2dz.
\end{align*}
We have, in terms of the variable $y$:
\begin{equation*}
  \(K\ast
  |\varphi^\eps|^2\)(t,x)=\int K\((y-z)\sqrt\eps\)\left\lvert
    u\(t,z\)\right\rvert^2dz.
\end{equation*}
This is where the smoothness of $K$ near the origin becomes important:
performing a Taylor expansion, we write
\begin{align*}
  K\((y-z)\sqrt\eps\)& = K(0) + \sqrt\eps (y-z)\cdot \nabla K(0) +
  \frac{\eps}{2} \<y-z,\nabla^2 K(0)(y-z)\> \\
&\quad + \eps^{3/2} r_K^\eps(y-z),
\end{align*}
with
\begin{equation*}
  |r_K^\eps(y-z)|\le C \<y-z\>^3,
\end{equation*}
for some $C$ independent of $\eps$, $y$ and $z$.  Therefore, using the
conservation of mass,
\begin{align*}
  \eps^{d/4}e^{-i\phi/\eps}&\(K\ast
  |\varphi^\eps|^2\)\varphi^\eps(t,x) = K(0)\|a\|_{L^2}^2 u(t,y) +
  \sqrt\eps \|a\|_{L^2}^2 y\cdot \nabla K(0)u(t,y)\\
&-\sqrt\eps \nabla K(0)\cdot G(t) u(t,y)
+\frac{\eps}{2}\<y,\nabla^2K(0)y\>\|a\|_{L^2}^2 u(t,y) \\
&+\frac{\eps}{2} \int \<z,\nabla^2K(0)z\>|u(t,z)|^2dz\times u(t,y)\\
&- \eps \<\nabla^2K(0)G(t),y\>u(t,y),
\end{align*}
where the notation $G(t)$ stands for
\begin{equation*}
  G(t)= \int_{\R^d}z|u(t,z)|^2dz.
\end{equation*}
We then discuss the outcome in \eqref{eq:generalstrategy} according to
the value of $\alpha$, on a formal level. We present the strategy to
justify the approximation in the case $\alpha=0$ only, since this case
contains all the arguments needed to treat the other cases. 

\subsection{Subcritical case: $\alpha >1$}
\label{sec:subsmooth}

When $\alpha>1$, by have $b_j=b_j^{\rm lin}$ for $j=0,1,2$, and 
\begin{equation*}
  r^\eps = \sqrt\eps r_V^\eps + \eps^{\alpha-1}\(K\ast
  |\varphi^\eps|^2\)\varphi^\eps.
\end{equation*}
Solving the equations $b_0^{\rm lin}=b_1^{\rm lin}=b_2^{\rm lin}$
leads to the approximate solution $\varphi^\eps_{\rm lin}$ defined in
Section~\ref{sec:linear}. 
\subsection{The critical case: $\alpha=1$}
\label{sec:alpha1}
When $\alpha=1$, we still have $b_j=b_j^{\rm lin}$ for $j=0,1$, but
the expression for $b_2$ is altered:
\begin{equation*}
  b_2 = i\d_t u +\frac{1}{2}\Delta u -
\frac{1}{2}\<y,\nabla^2 V\(t,x(t)\)y\>u-K(0)\|a\|_{L^2}^2 u.
\end{equation*}
The equation $b_2=0$ is the linear envelope equation \eqref{eq:linearu}, plus a
\emph{constant} potential, $K(0)\|a\|_{L^2}^2$. We infer
\begin{equation*}
  u(t,y)=u_{\rm lin}(t,y)e^{-i t K(0)\|a\|_{L^2}^2}.
\end{equation*}
The presence of this phase shift accounts for 
nonlinear effects at leading order in the approximate wave packet
$\varphi^\eps$. For the remainder term, we have:
\begin{align*}
  r^\eps(t,y) &= \sqrt\eps r_V^\eps(t,y) + \sqrt\eps \|a\|_{L^2}^2 
y\cdot \nabla K(0)u(t,y)\\
&-\sqrt\eps \nabla K(0)\cdot G(t) u(t,y)
+\frac{\eps}{2}\<y,\nabla^2K(0)y\>\|a\|_{L^2}^2 u(t,y) \\
&+\frac{\eps}{2} \int \<z,\nabla^2K(0)z\>|u(t,z)|^2dz\times u(t,y)\\
&- \eps \<\nabla^2K(0)G(t),y\>u(t,y),
\end{align*}
and we infer the (rough) pointwise estimate
\begin{equation}\label{eq:reste}
  |r^\eps(t,y)|\le C\sqrt\eps \<y\>^3|u(t,y)|\(1+ \|u(t)\|_{\Sigma}^2\).
\end{equation}

\subsection{A supercritical case: $\alpha=1/2$} 
\label{sec:alpha12}
For $\alpha<1$, we have to assume either $\alpha=1/2$ or $\alpha=0$ in
order to derive functions $b_j$ which do not depend on $\eps$. For the
simplicity of the presentation, we therefore stick to these cases, but
essentially, the case $1/2<\alpha<1$ is treated like the case
$\alpha=1/2$, and the case $0<\alpha<1/2$ like the case $\alpha=0$. 
\smallbreak

In the case $\alpha=1/2$, we still have $b_0=b_0^{\rm lin}$, but now
with
\begin{align*}
  b_1&= -i\( \dot x(t)-\xi(t)\)\cdot \nabla u-y\cdot \(\dot \xi(t)+
\nabla V\(t,x(t)\)\) u -K(0)\|a\|_{L^2}^2 u. \\
b_2 & = i\d_t u +\frac{1}{2}\Delta u -
\frac{1}{2}\<y,\nabla^2 V\(t,x(t)\)y\>u- \|a\|_{L^2}^2 y\cdot \nabla K(0) u+
\nabla K(0)\cdot G(t) u. 
\end{align*}
At this stage, it is easy to convince oneself that 
the equations $b_0=b_1=b_2=0$ are not compatible in general (if one
wants to consider a non-zero solution $u$). Therefore, we modify our
strategy, in order to allow $b_0$ to depend on $\eps$, so we can
upgrade the last factor in $b_1$ to $b_0$. This leads to:
\begin{align*}
  b_0^\eps&= -u\( \dot S(t)-\xi(t)\cdot \dot x(t)
  +\frac{|\xi(t)|^2}{2}+V\(t,x(t)\)+\sqrt\eps K(0)\|a\|_{L^2}^2\).\\
b_1&= -i\( \dot x(t)-\xi(t)\)\cdot \nabla u-y\cdot \(\dot \xi(t)+
\nabla V\(t,x(t)\)\) u .\\
b_2&= i\d_t u +\frac{1}{2}\Delta u -
\frac{1}{2}\<y,\nabla^2 V\(t,x(t)\)y\>u- \|a\|_{L^2}^2 y\cdot \nabla K(0) u+
\nabla K(0)\cdot G(t) u. 
\end{align*}
Keeping $(x(t),\xi(t))$ solution to the Hamiltonian flow
\eqref{eq:traj} leads to the $\eps$-dependent action:
\begin{equation*}
S^\eps(t) = \int_{0}^{t}\(\frac{1}{2}
    |\xi(s)|^2-V(s,x(s))\)ds-t\sqrt{\eps} K(0)\|a\|_{L^2(\R^d)}^2  .
\end{equation*}
The equation $b_1=0$ is then fulfilled as soon as we consider the
standard Hamiltonian flow. The equation $b_2=0$ is an envelope
equation, which is nonlinear since $G$ is a nonlinear function of
$u$. Note however that this yields a purely time-dependent potential.
Up to the time-dependent gauge
transform
\begin{equation*}
u(t, y)\mapsto u(t, y)\exp\(i\int_{0}^{t}
\nabla K(0)\cdot G(s)ds\),
\end{equation*}
which preserves the modulus of the unknown, hence $G$, the equation for
$u$ becomes a linear profile equation. Finally, we still have a
remainder term satisfying \eqref{eq:reste}. 

\subsection{Another supercritical case: $\alpha=0$} 
This case corresponds to the one studied in \cite{APPP-p}. We 
consider a more general framework though, since for instance we do not
assume that the kernel $K$ is radially symmetric. 
We now have
\begin{align*}
  b_0&= -u\( \dot S(t)-\xi(t)\cdot \dot x(t)
  +\frac{|\xi(t)|^2}{2}+V\(t,x(t)\)+K(0)\|a\|_{L^2}^2\), \\
b_1&=-i\( \dot x(t)-\xi(t)\)\cdot \nabla u-y\cdot \(\dot \xi(t)+
\nabla V\(t,x(t)\)\) u -
   \|a\|_{L^2}^2 y\cdot \nabla K(0)u\\
&\quad + \nabla K(0)\cdot G(t) u,\\
b_2&=  i\d_t u +\frac{1}{2}\Delta u- \frac{1}{2}\<y, M(t)y\>u +
\<\nabla^2K(0) G(t),y\>u \\
&\quad -\frac{1}{2} \int \<z,\nabla^2K(0)z\>|u(t,z)|^2dz\times u,
\end{align*}
where we have denoted 
\begin{equation*}\label{eq:M}
M(t)=\|a\|_{L^{2}(\R^{d})}^{2}\nabla^2 K(0)+\nabla^2_xV\(t,x(t)\).
\end{equation*}
Note that $M\in L^\infty_t(\R_+)$. We will assume
\begin{equation*}
  \nabla K(0)=0,
\end{equation*}
so $b_1=0$ as soon as $(x(t),\xi(t))$ satisfies \eqref{eq:traj}. 
 This assumption is a consequence of the framework in \cite{APPP-p},
 since the authors  
 suppose $K(x)=F(|x|)$ with $F$ even. Note that the slightly more
 general assumption $K(x)=K(-x)$ is physically relevant, in the sense
 that in that case, an energy can be associated to the Hartree
 nonlinearity (see e.g. \cite{CazCourant}):
 \begin{equation*}
   \iint K(x-y)\lvert \psi^\eps(t,x)\rvert^2\lvert\psi^\eps(t,y)\rvert^2dxdy.
 \end{equation*}
In the  case where $K$ is even, we obviously have $\nabla K(0)=0$. 
We then consider
 the Hamiltonian flow \eqref{eq:traj},  the modified  action
 \begin{equation}\label{eq:actionmodifiee}
   S(t) = \int_{0}^{t}\(\frac{1}{2}
    |\xi(s)|^2-V(s,x(s))\)ds-t K(0)\|a\|_{L^2(\R^d)}^2  ,
 \end{equation}
and the envelope equation $b_2=0$.
The remainder term still satisfies \eqref{eq:reste}.
\begin{remark}
Note that the Wigner measure of $\psi^{\eps}$ is not affected by the
nonlinearity:
\begin{equation*}
w(t, x, \xi)=\|a\|_{L^{2}(\R^{d})}^2\delta\(x-x(t)\)\otimes \delta
\(\xi-\xi(t)\),
\end{equation*}
in the four cases $\alpha>1$, $\alpha=1$, $\alpha=1/2$ and
$\alpha=0$, even though we have seen that the Hartree nonlinearity
does affect the leading order behavior of the wave function.
\end{remark}

\subsection{Sketch of the proof in the case $\alpha=0$}

We want to construct a solution to 
\begin{equation}
  \label{eq:usmooth}
  \begin{aligned}
     i\d_t u +\frac{1}{2}\Delta u&= \frac{1}{2}\<y, M(t)y\>u-
\<\nabla^2K(0) G(t),y\>u \\
&\quad + \frac{1}{2} \int \<z,\nabla^2K(0)z\>|u(t,z)|^2dz\times u, 
  \end{aligned}
\end{equation}
with initial datum $a$.
Introduce the solution to
\begin{equation}\label{eq:v}
  i\d_t v +\frac{1}{2}\Delta v = \frac{1}{2}\<y, M(t)y\>v-
\<\nabla^2K(0) {G}(t),y\>v \quad ;\quad v(0,y)=a(y).
\end{equation}
The functions $u$ and $v$ solve the same equation,
up to one term which can be absorbed by the 
gauge transform $u(t,y) =v(t,y) \exp\(i\theta (t)\)$, where
\begin{equation*}
\theta(t)= -\frac{1}{2}\int_{0}^{t}\int_{\R^{d}}
\<z,\nabla^2K(0)z\>|u(s,z)|^2dz ds.
\end{equation*}
Since this gauge transform does not affect the modulus of the
solution, one should consider that in \eqref{eq:v}, 
\begin{equation*}
  G(t) = \int_{\R^d}z|v(t,z)|^2dz.
\end{equation*}
\begin{lemma}\label{lem:alpha0}
  Let $k\ge 1$ and assume that $a\in \Sigma^k$. Then \eqref{eq:v} has a
  unique solution $v\in C(\R_+;\Sigma^k)$, and there exists $C$ such
  that
  \begin{equation*}
    \|v(t)\|_{\Sigma^k}\le Ce^{Ct},\quad t\ge 0. 
  \end{equation*}
As a consequence, \eqref{eq:usmooth} has a unique solution, which
possesses the same properties.
\end{lemma}
\begin{remark}
  The $\Sigma$ regularity is the least one has to demand in this
  result, for the gauge $\theta$ to be well defined (and the
  harmonic oscillator rotates the phase space, so the regularity must
  be the same in space and frequency). 
\end{remark}
\begin{proof}[Sketch of the proof]
  The main difficulty is that since the last term in the equation
  involves a time dependent potential which is unbounded in $y$, it
  cannot be treated as a perturbation. So to construct a local
  solution, we modify the standard Picard iterative scheme, to
  consider
\begin{equation}\label{eq:vn}
i\d_{t}v_{n}+\frac{1}{2}\Delta v_{n}=\frac{1}{2} \<y, M(t)y\>
v_{n}+\<\nabla ^{2} K(0){G}_{n-1}(t), y\>v_{n},\ n\ge  1,
\end{equation}
with $v_{n\mid t=0}=a$ for all $n$, $v_0(t,y)=a(y)$, and
\begin{equation*}
  G_{k}(t) = \int_{\R^d}z|v_k(s,z)|^2ds.
\end{equation*}
At each step, we solve  a linear equation, with a time dependent
potential which is at most quadratic: if $G_{n-1}\in L^\infty_{\rm
  loc}(\R_+)$, \cite{Fujiwara} ensures the existence of $v_n \in C(\R_+;
L^{2}(\R^{d}))$.  Applying the operators $y$ and $\nabla_y$ to
\eqref{eq:vn} shows that $ v_n \in C(\R_+;
\Sigma)$, hence $G_n \in L^\infty_{\rm  loc}(\R_+)$. To prove the
convergence of this scheme we need more precise (uniform in $n$)
estimates.
Direct computations show that 
\begin{equation*}
  \dot G_n(t) = \IM \int_{\R^d} \overline{v_n}(t,z)\nabla v_n(t,z)dz,
\end{equation*}
and
\begin{equation*}
\ddot G_{n}(t)+M(t)G_{n}(t)=\nabla^{2}K(0)
\|a\|_{L^{2}(\R^{d})}^{2}G_{n-1},\quad n\ge  1.
\end{equation*}
Let $f_n(t) = |\dot G_n(t)|^2 + |G_n(t)|^2$. We have
\begin{equation*}
  \dot f_n(t)\le 2|\dot G_n(t)\rvert\lvert \ddot G_{n}(t)| + 2|\dot
  G_n(t)\rvert\lvert  G_{n}(t)| \le C f_n(t) + C|G_{n-1}(t)|^2,
\end{equation*}
for some $C$ independent of $t$ and $n$. We infer that there exists
$C_0$ independent of $t\ge 0$ and $n$ such that
\begin{equation*}
  f_n(t) = |\dot G_n(t)|^2 + |G_n(t)|^2\le C_0 e^{C_0t}. 
\end{equation*}
By using energy estimates (applying the operators $y$ and $\nabla_y$
successively to the equation), we infer that there exists $C_1$
independent of $t\ge 0$ and $n$ such that 
\begin{equation*}
  \|v_n(t)\|_{\Sigma^k} \le C_1 e^{C_1 t}. 
\end{equation*}
The convergence of the sequence $v_n$ then follows: by a
standard fixed point argument, $v_n$ converges in $C([0,T];\Sigma)$ if
$T>0$ is sufficiently small. By using energy estimates, and
(exponential) \emph{a priori} bounds for $G$, we infer the exponential
control stated in the lemma, and hence global existence.
\end{proof}
\begin{remark}
  The above computations show that for $v$, the function $G$ satisfies
  \begin{align*}
    &\ddot G(t) + \nabla^2_xV\(t,x(t)\)G(t)=0,\\
&    G(0)=\int_{\R^d}z|a(z)|^2dz\quad ;\quad \dot G(0) = \IM
    \int_{\R^d}\overline a(z)\nabla a(z)dz.
  \end{align*}
In \cite{APPP-p}, the authors proved that if the initial data $a\in
\Sigma^{3}$  is such that
\begin{equation*}
  \int_{\R^d}z|a(z)|^2dz= \IM
    \int_{\R^d}\overline a(z)\nabla a(z)dz=0,
\end{equation*}
then $\int z |u(t,z)|^2dz=0$ for all time. The above ODE gives a
simple explanation of that property.
Note that up to changing $a$ to $b$ with 
\begin{equation*}
  b(y) = a(y-y_0)e^{iy\cdot \eta_0}
\end{equation*}
for $y_0$ and $\eta_0$ which can be computed explicitly,
that is up to a translation in the phase space, these two assumptions
are satisfied. 
However, the external potential is modified, and it is not so easy to
keep track of the geometric meaning of the approximation. This is why
we have chosen to sketch a direct approach here, which also shows that
\eqref{eq:v} is more nonlinear than it may seem. Note finally that
because of the term $G$, working in $L^2$ only would not be possible. 
\end{remark}
To conclude, we have:
\begin{proposition}\label{prop:smoothkernel}
Let $a\in \Sigma^3$, $\alpha=0$. Suppose $V$ satisfies
Assumption~\ref{hyp:V}, and $K$ satisfies
Assumption~\ref{hyp:smooth} and $\nabla K(0)=0$. Assume 
$\varphi^{\eps}$ is given by \eqref{eq:approx}, where the action is
given by \eqref{eq:actionmodifiee} and the envelope is given by
\eqref{eq:usmooth}. Then there exists a
positive constant  $C$ independent of $\eps$ such that
\begin{equation*}
  \|\psi^\eps(t)-\varphi^\eps(t)\|_{L^2(\R^d)}\le C\sqrt
  \eps e^{e^{C t}}, \quad t\ge 0.
\end{equation*}
In particular, there exists $c>0$ independent of $\eps$ such that
\begin{equation*}
 \sup_{0\le t\le c\ln\ln\frac{1}{\eps}}\|\psi^\eps(t)-\varphi^\eps(t)
 \|_{L^2(\R^d)}\Tend \eps 0 0 .
\end{equation*}
\end{proposition}
\begin{proof}
First, we change the unknown function $\psi^\eps$ to $u^\eps$ through
the bijective change of unknown function
\begin{equation*}
  \psi^\eps(t,x)=\eps^{-d/4} u^\eps 
\left(t,\frac{x-x(t)}{\sqrt\eps}\right)e^{i\left(S(t)+\xi(t)\cdot
    (x-x(t))\right)/\eps},
\end{equation*}
where $S$ is given by \eqref{eq:actionmodifiee}. Then \eqref{eq:NLS0}
(with $\alpha=0$) is equivalent to:
\begin{equation*}
    i\d_t u^\eps+\frac{1}{2}\Delta u^\eps = V^\eps u^\eps +
    \(K^\eps\ast |u^\eps|^2\)u^\eps\quad ;\quad u^\eps(0,y)=a(y),
\end{equation*}
where
\begin{align*}
  V^\eps(t,y)&= \frac{1}{\eps}\(V\(t,x(t)+y\sqrt\eps\) - V\(t,x(t)\) -
  \sqrt\eps y\cdot \nabla V\(t,x(t)\)\),\\
K^\eps(y)& = \frac{1}{\eps}\(K\(y\sqrt\eps\) -K(0)\). 
\end{align*}
We  have, in view of Assumption~\ref{hyp:V}, 
Assumption~\ref{hyp:smooth}, the property $\nabla K(0)=0$, and
Taylor's formula, 
\begin{equation*}
  |V^\eps(t,y)| +|K^\eps(y)|\le C |y|^2,
\end{equation*}
for some constant $C$ independent of $\eps$, $t$ and $y$. We already
know from Lemma~\ref{lem:smoothglobal} that $u^\eps\in C(\R_+;L^2)$,
with $\|u^\eps(t)\|_{L^2}=\|a\|_{L^2}$. Proceeding in the same way as
in the proof of Lemma~\ref{lem:alpha0}, we can also prove that $u^\eps
\in C(\R_+;\Sigma^3)$ and that there exists $C$ such that
\begin{equation*}
  \|u^\eps(t)\|_{\Sigma^3}\le Ce^{Ct}. 
\end{equation*}
Set $w^\eps = u^\eps-u$: we have
$\|w^\eps(t)\|_{L^2}=\|\psi^\eps(t)-\varphi^\eps(t)\|_{L^2}$, and
$w^\eps$ solves
\begin{equation*}
  i\d_t w^\eps +\frac{1}{2}\Delta w^\eps = V^\eps w^\eps +\(K^\eps\ast
  |u^\eps|^2\)u^\eps - \(K^\eps\ast |u|^2\)u - r^\eps\quad ;\quad
  w^\eps_{\mid t=0}=0,
\end{equation*}
 where $r^\eps$ satisfies the pointwise estimate
 \eqref{eq:reste}. Write
 \begin{equation*}
  \(K^\eps\ast
  |u^\eps|^2\)u^\eps - \(K^\eps\ast |u|^2\)u =  \(K^\eps\ast
  |u^\eps|^2\) w^\eps +\(K^\eps\ast \(|u^\eps|^2-|u|^2\)\)u,
 \end{equation*}
so the standard $L^2$ estimate yields
\begin{equation*}
  \|w^\eps(t)\|_{L^2}\le \int_0^t\left\lVert \(K^\eps\ast
    \(|u^\eps|^2-|u|^2\)\)u(s)\right\rVert_{L^2}ds + 
  \int_0^t\|r^\eps(s)\|_{L^2}ds. 
\end{equation*}
Since $|u^\eps|^2-|u|^2= 2\RE (\overline u w^\eps)+|w^\eps|^2$, and
$\|u^\eps(t)\|_{\Sigma^2}+ \|u(t)\|_{\Sigma^3}\le Ce^{Ct}$, we come up
with
\begin{equation*}
  \|w^\eps(t)\|_{L^2}\le C\int_0^t e^{Cs}\|w^\eps(s)\|_{L^2}ds +
  C\sqrt\eps \int_0^t e^{Cs}ds.
\end{equation*}
Gronwall lemma yields
\begin{equation*}
  \|w^\eps(t)\|_{L^2}\le C \sqrt\eps e^{e^{Ct}},
\end{equation*}
and the result follows.
\end{proof}
\begin{remark}
  It is quite surprising that even in a supercritical case, the
  approximation can be proven so simply, eventually by a Gronwall type
  argument. This is in sharp contrast with supercritical WKB analysis
  for the nonlinear Schr\"odinger equation (see \cite{CaBook}). On the
  other hand, we had to use the \emph{a priori} control of the
  approximate solution $u$ \emph{and} of the exact solution $u^\eps$:
  in subcritical or critical cases, controlling the approximate
  solution is sufficient in general, as illustrated below in the
  case of homogeneous kernels. 
\end{remark}
\begin{remark}
  In the cases $\alpha=1/2$, $\alpha=1$, and $\alpha>1$, a similar statement can
  be proved, and the time of validity is improved, in the sense that
  the $c\ln\ln 1/\eps$ in the end of
  Proposition~\ref{prop:smoothkernel} can be replaced by $c\ln
  1/\eps$. More precisely, the error estimate will be:
  \begin{itemize}
  \item $\alpha>1$: $\|\psi^\eps(t)-\varphi^\eps_{\rm
      lin}(t)\|_{L^2}\le C\eps^{\kappa}e^{Ct}$
  with $\kappa= \min(1/2,\alpha-1)$. 
\item $\alpha =1$: $\|\psi^\eps(t)-\varphi^\eps(t)\|_{L^2}\le
  C\sqrt\eps e^{Ct}$, with $\varphi^\eps$ as in \S\ref{sec:alpha1}.
\item $\alpha =\frac{1}{2}$: $\|\psi^\eps(t)-\varphi^\eps(t)\|_{L^2}\le
  C\sqrt\eps e^{Ct}\exp\(\sqrt\eps e^{Ct}\)$, with $\varphi^\eps$
  as in \S\ref{sec:alpha12}. 
  \end{itemize}
\end{remark}

\section{Homogeneous kernel: technical background}

 In this section, we present some general technical tools which will
 be used in the proofs of the main results in the case of an
 homogeneous kernel. In particular, we
establish the global well-posedness for \eqref{eq:nlhartree}, and
estimate the evolution of weighted Sobolev norms of the solution  over
large time: 

\begin{proposition}\label{prop:global}
Let $\lambda\in \R$ and $0<\gamma<\min(2,d).$ Suppose $V$ satisfies
Assumption~\ref{hyp:V} and the initial data $a\in L^2(\R^d)$. Then
there exists a unique solution $u\in C(\R_{+}; L^2(\R^d))\cap
L_{{\rm loc}}^{8/\gamma}(\R_{+}, L^{4d/(2d-\gamma)}(\R^{d}))$ to
\eqref{eq:nlhartree}. If in addition $a\in \Sigma^k$ for some $k\in
\N$, then $u\in C(\R_{+};\Sigma^k)$, and there exists $C=C(k)$ such
that
\begin{equation*}
  \|u(t)\|_{\Sigma^k}\le C e^{Ct},\quad \forall t\ge 0.
\end{equation*}
\end{proposition}

\subsection{Strichartz estimates}\label{sec:Strichartz}

Before studying the semi-classical limit, we recall some known
facts and establish technical results.
\begin{definition}\label{def10}
A pair $(q,r)$ is admissible if $2\le r<\frac{2d}{d-2}$ 
($2\le q\le\infty$ if $d=1$, $2\le q<
  \infty$ if $d=2$) and
\begin{equation*}
\frac{2}{q}=d\(\frac{1}{2}-\frac{1}{r}\):=\delta(r).
\end{equation*}
\end{definition}
Following \cite{GV85,KT,Yajima87}, Strichartz estimates are
available for the Schr\"odinger equation without external potential.
Thanks to the construction of the parametrix performed in
\cite{Fujiwara79,Fujiwara}, similar results are available in the
presence of an external potential satisfying
\begin{hyp}\label{hyp:W}
$W\in L_{{\rm loc}}^{\infty}(\R_{+}\times \R^{d})$ is a smooth
with respect to $x$ for all $t\ge  0$: $x\mapsto W(t,x)$ is
a $C^{\infty}$ map. Moreover, it is subquadratic in $x$:
\begin{equation*}
\forall  \beta\in \N^{d},\  |\beta| \ge  2,\
\partial_{x}^{\beta} W \in L^{\infty}(\R_+\times\R^{d}).
\end{equation*}
\end{hyp}
 Define
 $U^{\varepsilon}(t,s)$ the semigroup as $u^\eps(t, x)= U^\eps(t, s)\phi(x)$,
 where
 \begin{equation*}
 i\eps\d_{t}u^\eps+\frac{\eps^2}{2}\Delta u^\eps= W(t, x)u^\eps;\quad
 u^\eps(s,x)=\phi(x). 
 \end{equation*}
From \cite{Fujiwara79, Fujiwara}, it has the following properties:
 \begin{itemize}
 \item $U^\eps(t,t)={\rm Id}.$
   \item The map $(t, s)\mapsto U^{\varepsilon}(t,s)$ is strongly
   continuous.
   \item $U^{\varepsilon}(t,\tau)U^{\varepsilon}(\tau, s)=U^{\varepsilon}(t, s).$
   \item
   $U^{\varepsilon}(t, s)^{*}=U^{\varepsilon}(t, s)^{-1}.$
   \item $U^{\varepsilon}(t, s)$ is unitary on $L^{2}$ norm:
     $\|U^{\varepsilon}(t,
     s)\phi\|_{L^{2}(\R^{d})}=\|\phi\|_{L^{2}(\R^{d})}.$ 
   \item There exist $\delta, C>0$ independent of $\varepsilon\in
   ]0,1]$ such that for all $t,s\ge 0$ with $|t-s|<\delta$,
   \begin{equation*}
   \|U^{\varepsilon}(t,0)U^{\eps}(s, 0)^{*}\phi\|_{L^{\infty}(\R^{d})}\le
   \frac{C}{(\varepsilon|t-s|)^{d/2}}\|\phi\|_{L^{1}(\R^{d})}.
   \end{equation*}
 \end{itemize}
 Scaled Strichartz
 estimates follow from the above dispersive relation:
 \begin{proposition}[Scaled Strichartz
   estimates]\label{prop:strichartz}
 Let $U^{\varepsilon}(t, s)$ defined as above. There exists
 $\delta>0$ independent of $\eps$ such that the following holds. \\
\noindent $(1)$  For any admissible pair $(q,r)$, there exists $C(q)$
 independent of $\varepsilon$ such that
\begin{equation*}
\varepsilon^{1/q}\|U^{\varepsilon}(\cdot, s)\phi\|_{L^{q}([s,s+\delta];
L^{r}(\R^{d}))}\le C(q)\|\phi\|_{L^{2}(\R^{d})}, \quad \forall
\phi\in L^{2}(\R^{d}),\quad \forall s\ge 0.
\end{equation*}
$(2)$ For $s\in \R$, denote
\begin{equation*}
D_{s}^{\varepsilon}(F)(t,x)=\int_{s}^{t}U^{\varepsilon}(t,
s)F(s,x)ds,
\end{equation*}
and $I=[s, s+\eta]$. For all admissible pairs $(q_{1}, 
 r_{1})$ and $(q_{2}, r_{2})$, there exists $C(q_{1}, q_{2})$ independent of
$\varepsilon$ and $s\ge 0$ such that
\begin{equation*}
\varepsilon^{1/q_{1}+1/q_{2}}\|D^\eps_{s}(F)\|_{L^{q_{1}}(I;
L^{r_{1}}(\R^{d}))}\le C(q_{1}, q_{2})\|F\|_{L^{q_{2}^{\prime}}(I;
L^{r_{2}^{\prime}}(\R^{d}))}
\end{equation*}
for all $F\in L^{q_{2}^{\prime}}(I; L^{r_{2}^{\prime}}(\R^{d}))$ and
$0\le \eta \le \delta.$ Here
$\frac{1}{q_{2}}+\frac{1}{q_{2}^{\prime}}=1$ and
$\frac{1}{r_{2}}+\frac{1}{r_{2}^{\prime}}=1.$
\end{proposition}

This statement will be used in the two cases $\eps=1$ (for the
envelope equation \eqref{eq:nlhartree}), and $\eps \in ]0,1]$ (to justify
the approximation of the exact solution $\psi^\eps$). 


\subsection{Global existence in $L^{2}$}
\label{sec:profile} 
We consider a rather general
potential $W$ satisfying Assumption ~\ref{hyp:W}, and 
consider the Cauchy problem 
on $\R_{+}\times\R^{d}:$
\begin{equation}\label{eq:genlhartree}
 i\partial_{t} v+\frac{1}{2} \Delta v= W(t, y) v+\lambda
|y|^{-\gamma}*|v|^{2} v\quad ;\quad
     v|_{t=s}=v_{s}.
\end{equation}
This form includes the cases of the exact solution $\psi^\eps$ in
\eqref{eq:NLS0} as well as the envelope equation
\eqref{eq:nlhartree}. We establish 
global existence for ~\eqref{eq:genlhartree} in the
$L^{2}$-subcritical case $0<\gamma<\min \(2, d\)$, yielding the first part of
Proposition~\ref{prop:global}. 
\begin{lemma}\label{lem:localhartree}
Let $\lambda\in \R$, $0<\gamma<\min\(2,d\)$. Suppose $W$
satisfies Assumption ~\ref{hyp:W}. Then for $s=0$ and $v_{0}\in
L^{2}(\R^{d})$, ~\eqref{eq:genlhartree} has a unique solution
\begin{equation*}
v\in C(\R_{+}; L^{2}(\R^{d}))\cap L_{{\rm loc}}^{8/\gamma}(\R_{+},
L^{4d/(2d-\gamma)}(\R^{d})).
\end{equation*}
In addition, the $L^{2}$ norm of $v$ is conserved:
\begin{equation*}
\|v(t)\|_{L^{2}(\R^{d})}=\|v_{0}\|_{L^{2}(\R^{d})}, \ \forall t\ge
0.
\end{equation*}
\end{lemma}

\begin{proof}
By Duhamel's formula, we write ~\eqref{eq:genlhartree} as
\begin{equation*}
v(t)=U(t,s)v_s-
i\lambda\int_{s}^{t}U(t,\tau)(|x|^{-\gamma}*|v|^{2}v)(\tau)d\tau=:\Phi^{s}(v)(t), 
\end{equation*}
where we have dropped the dependence of $U^\eps$ upon $\eps$ in the
notation, since we assume $\eps=1$ here. 
 Introduce the space
 \begin{equation*}
\begin{aligned}
Y_{s,T}&=\{\phi\in C(I_{s, T}; L^{2}(\R^{d})):
\|\phi\|_{L^{\infty}(I_{s, T};\ L^{2}(\R^{d}))}\le
2\|v_{s}\|_{L^{2}(\R^{d})},\\
&\quad \|\phi\|_{L^{8/\gamma}(I_{s, T};\
L^{4d/(2d-\gamma)}(\R^{d}))}\le 2C(8/\g)\|v_{s}\|_{L^{2}(\R^{d})}\},
\end{aligned}
\end{equation*}
and the distance $$ d(\phi_{1},
\phi_{2})=\|\phi_{1}-\phi_{2}\|_{L^{8/\gamma}(I_{s, T};\
L^{4d/(2d-\gamma)})},$$ 
where $I_{s, T}=[s, s+T]$ with $s\ge 0$ and
$T>0$, and $C(8/\g)$ stems from Proposition~\ref{prop:strichartz}.  
Then $(Y_{s,T}, d)$ is a Banach space,
  as remarked in \cite{Kato87} (see also \cite{CazCourant}).  
Hereafter, we denote by
\begin{equation*}
q=\frac{8}{\gamma}, \quad r=\frac{4d}{2d-\gamma}, \quad
\theta=\frac{8}{4-\gamma},
\end{equation*}
and $\|\cdot\|_{_{L^{a}(I_{s,T};\ L^{b}(\R^{d}))}}$ by
$\|\cdot\|_{L^{a}L^{b}}$ for simplicity. Notice that $\(q, r\)$ is
admissible and
\begin{equation*}
\frac{1}{q^{\prime}}=\frac{4-\gamma}{4}+\frac{1}{q}=\frac{1}{2}+\frac{1}{\theta}
\quad ; \quad
\frac{1}{r^{\prime}}=\frac{\gamma}{2d}+\frac{1}{r}\quad;\quad
\frac{1}{2}=\frac{1}{\theta}+\frac{1}{q}.
\end{equation*}
 By using Strichartz estimates, H\"older inequality and
 Hardy--Littlewood--Sobolev inequality, we have, for $(\underline
 q,\underline r)\in \{(q,r),(\infty,2)\}$:
\begin{equation}\label{4.2}
\begin{aligned}
\|\Phi^{s}(v)\|_{L^{\underline q}L^{\underline r}}&\le C(\underline q)
\|v_{s}\|_{L^{2}}+C(\underline q,q)\left\lVert |y|^{-\gamma}*|v|^{2}v
\right\rVert_{L^{q'}L^{r'}}\\
&\le C(\underline q)\|v_{s}\|_{L^{2}}+
C(\underline q,q)\left\lVert
  |y|^{-\gamma}*|v|^{2}\right\rVert_{L^{4/(4-\gamma)}L^{2d/\gamma}}
\|v\|_{L^{q}L^{r}}\\
& \le C(\underline q)
\|v_{s}\|_{L^{2}}+C\|v\|_{L^{\theta}L^{r}}^{2}\|v\|_{L^{q}L^{r}}\\
&\le C(\underline q)\|v_{s}\|_{L^{2}}+CT^{1-\gamma/2}\|v\|_{L^{q}L^{r}}^{3},
\end{aligned}
\end{equation}
 for any $v\in Y_{s,T}$, with $C(\infty)=1$ by the standard energy
 estimate. To show the contraction property
of $\Phi^{s}$,  for any $v, w\in
Y_{s,T}$, we get
\begin{align*}
\|\Phi^{s}(v)-\Phi^{s}(w)\|_{L^{q}L^{r}}
&\lesssim\||y|^{-\gamma}*|v|^{2}\|_{L^{4/(4-\gamma)}L^{2d/\gamma}}
\|v-w\|_{L^{q}L^{r}}\\ 
&\quad
+\left\lVert |y|^{-\gamma}*\left\lvert |v|^{2}-|w|^{2}\right\rvert
\right\rVert_{L^{2}L^{2d/\gamma}}\|w\|_{L^{\theta}L^{r}}\\
&\lesssim
\(\|v\|_{L^{\theta}L^{r}}^{2}+\|w\|_{L^{\theta}L^{r}}^{2}\)\|v-w\|_{L^{q}L^{r}}\\
& \le
C T^{1-\gamma/2}(\|v\|_{L^{q}L^{r}}^{2}+\|w\|_{L^{q}L^{r}}^{2})\|v-w\|_{L^{q}L^{r}}.
\end{align*}
 Thus $\Phi^{s}$ is a contraction from $Y_{s,T}$ to $Y_{s, T}$ provided
that $T$ is sufficiently small. Then there exists  a unique $v\in Y_{s,
T}$ solving \eqref{eq:genlhartree}. The global existence of the
solution for \eqref{eq:genlhartree} follows from the conservation
of $L^{2}$-norm of $v$. 
\end{proof}

\subsection{Growth of higher order Sobolev norms and momenta}
We now consider \eqref{eq:nlhartree}, that is 
\begin{equation*}
    i\partial_{t}u+\frac{1}{2}\Delta u=\frac{1}{2}\<y, Q(t)y\>u+
\lambda |y|^{-\gamma}*|u|^{2}u\quad;
    \quad
    u|_{t=0}=a,
\end{equation*}
where $Q(t) = \nabla^2 V\(t,x(t)\)$, so $Q\in C(\R_{+}, \R)$ is
locally Lipschitzean, and bounded, by 
Assumption~\ref{hyp:V}. The second 
part of Proposition ~\ref{prop:global} follows from
 the following lemma.
\begin{lemma}\label{lem:exphartree}
 Let $\lambda\in
\R$ and $0<\gamma<\min\(2,d\)$.  Suppose $a \in\Sigma^{k}$ for some
$k\in \N$. Then there exists a unique $u\in C(\R_{+}; \Sigma^{k})$
solving \eqref{eq:nlhartree}, and there exists $C=C(k)$ such that for
every admissible pair $(q_{1}, r_{1})$,
\begin{equation}\label{eq:growthu}
  \|y^{\alpha}\d_{y}^{\beta}u(t)\|_{L^{q_{1}}([0, t];
    L^{r_{1}}(\R^{d}))}
\le C e^{Ct},\quad \forall t\ge
0,\  \alpha,\beta\in \N^{d} , |\alpha|+|\beta|\le k.
\end{equation}
\end{lemma}
\begin{proof}
We just state the proof of ~\eqref{eq:growthu}, by borrowing the
approach in \cite{Ca-p}. Applying similar arguments as the proof of
Lemma ~\ref{lem:localhartree} and induction, one can prove
global existence and uniqueness of the $\Sigma^{k}$
solution for ~\eqref{eq:nlhartree}.\smallbreak
 \textbf{Step 1:} $k=0$. For
all $t\ge  0$ and $\tau>0$, set $I=[t,t+\tau]$. Resuming the
computations as in ~\eqref{4.2}, we have
\begin{equation*}
\|u\|_{L^{q}(I;L^{r}(\R^{d}))\cap L^{\infty}(I; L^{2}(\R^{d}))}\le
C\|u(t)\|_{L^{2}}+C_{1}\tau^{1-\gamma/2}\|u\|_{L^{q}(I;
L^{r}(\R^{d}))}^{3},
\end{equation*}
where $C$ and $C_{1}$ is independent of $t$ and $\tau$. Then
\eqref{eq:growthu} for $k=0$ follows from the following bootstrap argument.
\begin{lemma}[Bootstrap argument]
Let $g=g(t)$ be a nonnegative continuous function on $[0,T]$ such
that for every $t\in [0,T]$,
\begin{equation*}
g(t)\le M+\delta g(t)^{\kappa}
\end{equation*}
where $M, \delta$ and $\kappa>1$ are constants such that
\begin{equation*}
M<\(1-\frac{1}{\kappa}\)\frac{1}{(\kappa\delta)^{1/(\kappa-1)}};\
\ g(0)\le \frac{1}{(\kappa\delta)^{1/(\kappa-1)}}.
\end{equation*}
Then
\begin{equation*}
g(t)\le \frac{\kappa}{\kappa-1}M,\quad \forall t\in [0,T].
\end{equation*}
\end{lemma}
For fixed $\tau$ small enough, by the conservation of $L^{2}$ norm
 of $u$, we choose $\kappa=3$ and
$\delta=C_{1}\tau^{1-\gamma/2}$. Since $0<\gamma<2,$ at every
increment of time of length of $\tau$, the $L^{q}L^{r}$ and
$L^{\infty}L^{2}$ norms of $u$ are bounded by
$\frac{3}{2}C\|u_{0}\|_{L^{2}}$, and  \eqref{eq:growthu} follows in
the case $k=0$, $(q_1,r_1)=(q,r)$. Using Strichartz inequalities again,
we have, for any admissible pair $(q_1,r_1)$,
\begin{equation*}
\|u\|_{L^{q_1}(I;L^{r_1}(\R^{d}))}\lesssim
\|u(t)\|_{L^{2}}+\tau^{1-\gamma/2}\|u\|_{L^{q}(I;
L^{r}(\R^{d}))}^{3},
\end{equation*}
and \eqref{eq:growthu} follows in the case $k=0$.
\smallbreak
 \textbf{Step 2:} Suppose Lemma ~\ref{lem:exphartree} holds for
 ${k-1}$ $(k \ge
1)$.  We denote by $w_{\ell}$ the family of combination of $\alpha$
momenta and $\beta$ space derivatives of $u$ with
$|\alpha|+|\beta|=\ell$. Applying 
$y^{\alpha}\partial_{y}^{\beta}$ to \eqref{eq:nlhartree} formally, we obtain
\begin{align*}
i\partial_{t}w_{k}+\frac{1}{2}\Delta w_{k}&=\frac{1}{2}\<y,
Q(t) y\>w_{k}+{\tt V}(u, w_{k})+L_k(u)\\
&\quad +\lambda\sum_{{0\le
j_{1},j_{2},j_{3}\le
k-1}\atop {j_{1}+j_{2}+j_{3}=k}}
c_{j_1,j_2,j_3}|y|^{-\gamma}*(w_{j_{1}}\overline{w}_{j_{2}})w_{j_{3}},
\end{align*}
where 
\begin{align*}
& {\tt V}(u,
w_{k})=|y|^{-\gamma}*(w_{k}\overline{u})u+|y|^{-\gamma}*(\overline{w}_{k}u)u
+|y|^{-\gamma}*|u|^{2}w_{k},\\
& L_k(u)=\frac{1}{2}\Omega(t)\left[y^{\alpha}\partial_{y}^{\beta},
|y|^{2}\right]u+\frac{1}{2}\left[\Delta, y^{\alpha}\partial_{y}^{\beta}\right]u.
\end{align*}
Notice that $L_k(u)$ is controlled pointwise by $w_k$. Still by
 Strichartz estimates, induction
 and Step 1, we have
\begin{align*}
\|w_{k}\|_{L^{q}(I; L^{r})\cap L^\infty(I;L^2)}
&\lesssim
\|w_{k}(t)\|_{L^{2}}+\tau^{1-\frac{\gamma}{2}}\|u\|_{L^q(I;L^r)}\|w_{k}\|_{L^{q}(I;
L^{r})}\\
&\quad +C_{2}\|w_{k}\|_{L^{1}(I,L^{2})}+C_{3}
e^{3C(t+\tau)}\\
&\le
C\|w_{k}(t)\|_{L^{2}}+C_1\tau^{1-\frac{\gamma}{2}}\|w_{k}\|_{L^{q}(I;
L^{r})}\\
&\quad +C_{2}\|w_{k}\|_{L^{1}(I,L^{2})}+C_{3}
e^{3C(t+\tau)},
\end{align*}
where $C,C_{1}, C_{2}$ and $C_{3}$ are independent of $t$ and $\tau$.
Choosing $\tau<1$ fixed small enough, the second term on the right
hand side of the above inequality can be absorbed by the left hand
side. 

 For
any time interval $[0, t]$, split it into finitely many pieces such
that the length of every piece at most $\tau$, then we have
\begin{equation*}
\|w_{k}\|_{L^{q}([0, t]; L^{r}(\R^{d}))\cap L^{\infty}([0,
t];L^{2}(\R^{d}))}\lesssim \|w_{k}\|_{L^{1}([0,
t];L^{2}(\R^{d}))}+\int_{0}^{t}e^{3Cs}ds.
\end{equation*}
Lemma ~\ref{lem:exphartree} follows from the Gronwall lemma in the
case $(q_1,r_1)=(\infty,2)$. Using Strichartz inequalities again, the
general case follows.
\end{proof}

\section{Bounded time interval for the critical case}
\label{sec:bounded}

In this section, we consider the critical case for ~\eqref{eq:NLS0}
with homogeneous nonlinearity in bounded time interval, and establish
a good approximation to the wave function. 
First, we recall the following lemma from \cite{CaFe11}.
\begin{lemma}\label{lem:operators}
Suppose $V$ satisfies Assumption ~\ref{hyp:V}. Let $(x(t), \xi(t))$
be defined by the trajectories ~\eqref{eq:traj} and $S(t)$ be the
classical action ~\eqref{eq:classicalaction}. Assume
$A^{\varepsilon}$ and $B^{\varepsilon}$ are defined as follows:
\begin{align*}
&
A^{\varepsilon}=\sqrt{\varepsilon}\nabla-
i\frac{\xi(t)}{\sqrt{\varepsilon}}=\sqrt{\varepsilon}
e^{i(S(t)+\xi(t)\cdot(x-x(t)))/\eps}\nabla
\(e^{-i(S(t)+\xi(t)\cdot(x-x(t)))/{\eps}}\cdot\);\\
& B^{\varepsilon}=\frac{x-x(t)}{\sqrt{\varepsilon}}.
\end{align*}
Then $A^{\varepsilon}$ and $B^{\varepsilon}$ satisfy the commutation
relations:
\begin{align*}
& \left[i\varepsilon\partial_{t}+\frac{\varepsilon^{2}}{2}\Delta-V,\ \
A^{\varepsilon}\right]=\sqrt{\varepsilon}(\nabla V(t, x)-\nabla V(t, x(t)));\\
& \left[i\varepsilon\partial_{t}+\frac{\varepsilon^{2}}{2}\Delta-V,\ \
B^{\varepsilon}\right]=\varepsilon A^{\varepsilon}.
\end{align*}
\end{lemma}
\begin{proposition}\label{prop:bounded}
Under the assumptions in Theorem ~\ref{thm:criticalhartree}, for all
$T>0$ which is independent of $\eps>0$, we have
\begin{equation*}
\sup_{0\le t\le
T}\|\psi^{\eps}(t)-\varphi^{\eps}(t)\|_{\mathcal{H}}=\mathcal{O}(\sqrt{\eps}).
\end{equation*}
\end{proposition}
\begin{proof}
 Set
$w^{\varepsilon}=\psi^{\varepsilon}-\varphi^{\varepsilon}$: it satisfies
\begin{equation}\label{5.4}
    i\varepsilon\partial_{t}w^{\varepsilon}+\frac{\varepsilon^{2}}{2}\Delta
    w^{\varepsilon}=Vw^{\varepsilon}-L^{\varepsilon}+N^{\varepsilon}\quad;
    \quad
    w^{\varepsilon}|_{t=0}=0.
\end{equation}
We have denoted by
\begin{align*}
& L^{\varepsilon}=\(V(t, x)-T_{2}(t, x, x(t))
\)\varphi^{\varepsilon},\\
& N^{\varepsilon}=\lambda
\varepsilon^{\alpha_{c}}\(|x|^{-\gamma}*|\psi^{\varepsilon}|^{2}\psi^{\varepsilon}-
|x|^{-\gamma}*|\varphi^{\varepsilon}|^{2}\varphi^{\varepsilon}\),
\end{align*}
where $T_{2}$ corresponds to a second order Taylor approximation:
\begin{align*}
T_{2}(t, x, x(t))&=V\(t, x(t)\)+\<\nabla V(t, x(t)),\  x-x(t)\>\\
&\quad +\frac{1}{2}\<x-x(t), \nabla^{2}V(t, x(t))(x-x(t))\>.
\end{align*}
By Duhamel's formula, using scaled
 Strichartz estimates and similar arguments as in the proof of Lemma
 ~\ref{lem:localhartree}, we have
 \begin{equation}\label{5.7}
\begin{aligned}
\|w^{\varepsilon}\|_{L^{q}L^{r}}& \lesssim
\varepsilon^{-1/q}\|w^{\varepsilon}(t)\|_{L^{2}}+\varepsilon^{-1-1/q}\|L^{\varepsilon}\|_{L^{1}L^{2}}
+\varepsilon^{-1-2/q}\|N^{\varepsilon}\|_{L^{q^{\prime}}L^{r^{\prime}}}\\
& \lesssim
\varepsilon^{-1/q}\|w^{\varepsilon}(t)\|_{L^{2}}+\varepsilon^{-1-1/q}\|L^{\varepsilon}\|_{L^{1}L^{2}}\\
&
\quad+\varepsilon^{\alpha_{c}-1-2/q}(\|\psi^{\varepsilon}\|_{L^{\theta}L^{r}}^{2}
+\|\varphi^{\varepsilon}\|_{L^{\theta}L^{r}}^{2})\|w^{\varepsilon}\|_{L^{q}L^{r}}.
 \end{aligned}
 \end{equation}
We have used the notation $L^{a}L^{b}$ for $L^{a}([t, t+\tau];
L^{b}(\R^{d}))$ with $t\ge 0$ and $\tau>0$. For all $T>0$, using
similar arguments as in the proof of Lemma ~\ref{lem:localhartree},
we know $u^{\varepsilon}$, $u\in C(\R_{+}; \Sigma)$, then
 \begin{equation*}
 \|Pu\|_{L^{\infty}([0,T];
 L^{2}(\R^{d}))}+\|Pu^{\varepsilon}\|_{L^{\infty}([0,T];
 L^{2}(\R^{d}))}\le C(T),\quad \forall P\in\{{\rm Id}, \nabla ,
 x\}.
 \end{equation*}
 By the definitions of $A^{\eps}$ and $B^{\eps}$, it is easy to
 check
\begin{equation*}
\|P^{\varepsilon}\varphi^{\varepsilon}\|_{L^{\infty}([0,T];
L^{2}(\R^{d}))}+\|P^{\varepsilon}\psi^{\varepsilon}\|_{L^{\infty}([0,T];
L^{2}(\R^{d}))}\le C(T),\quad \forall P^{\varepsilon}\in\{{\rm Id},
A^{\varepsilon}, B^{\varepsilon}\},
\end{equation*}
where $C(T)$ is independent of $\varepsilon.$ 
Gagliardo--Nirenberg inequality then yields
\begin{equation*}
\|\varphi^{\varepsilon}(t)\|_{L^{r}(\R^{d})}\lesssim
\varepsilon^{-\gamma/8}\|\varphi^{\varepsilon}\|_{L^{2}(\R^{d})}^{1-\gamma/4}
\|A^{\varepsilon}\varphi^{\varepsilon}\|_{L^{2}(\R^{d})}^{\gamma/4}\le
C(T)\varepsilon^{-\gamma/8}, \quad \forall t\in [0,T].
\end{equation*}
Similarly, $\|\psi^{\varepsilon}(t)\|_{L^{r}}$ is bounded by
$C(T)\varepsilon^{-\gamma/8}$ for all $t\in [0, T]$, where $C(T)$ is
independent of $\varepsilon.$ Let $[t,t+\tau]\subset [0,T],$ we can
rewrite \eqref{5.7} as
\begin{equation}\label{5.8}
\begin{aligned}
& \|w^{\varepsilon}\|_{L^{q}L^{r}}\lesssim
\varepsilon^{-1/q}\|w^{\varepsilon}(t)\|_{L^{2}}
+\varepsilon^{-1-1/q}\|L^{\varepsilon}\|_{L^{1}L^{2}}\\
&\quad+\varepsilon^{\alpha_{c}-1-2/q}
\tau^{1-\gamma/4}(\|\psi^{\varepsilon}\|_{L^{\infty}L^{r}}^{2}
+\|\varphi^{\varepsilon}\|_{L^{\infty}L^{r}}^{2})\|w^{\varepsilon}\|_{L^{q}L^{r}}\\
&\lesssim \varepsilon^{-1/q}\|w^{\varepsilon}(t)\|_{L^{2}}
+\varepsilon^{-1-1/q}\|L^{\varepsilon}\|_{L^{1}L^{2}}+\varepsilon^{\alpha_{c}-1-2/q
-\gamma/4}\tau^{1-\gamma/4}\|w^{\varepsilon}\|_{L^{q}L^{r}}\\
&\lesssim \varepsilon^{-1/q}\|w^{\varepsilon}(t)\|_{L^{2}}
+\varepsilon^{-1-1/q}\|L^{\varepsilon}\|_{L^{1}L^{2}}+\tau^{1-\gamma/4}
\|w^{\varepsilon}\|_{L^{q}L^{r}}.
\end{aligned}
\end{equation}
Choosing $\tau$ sufficiently small,  the
last term on the right hand side in ~\eqref{5.8} can be absorbed by
the left hand side.
Splitting $[0, T]$ into finitely many such intervals, we have
\begin{equation*}
\|w^{\varepsilon}\|_{L^{q}([0,T]; L^{r}(\R^{d}))}\lesssim
\varepsilon^{-1/q}\|w^{\varepsilon}\|_{L^{1}([0,T];
L^{2}(\R^{d}))}+\varepsilon^{-1-1/q}\|L^{\varepsilon}\|_{L^{1}([0,T];
L^{2}(\R^{d}))}.
\end{equation*}
Using Strichartz estimates again and resuming the above
computations, we get
\begin{equation*}
\|w^{\varepsilon}\|_{L^{\infty}([0,t]; L^{2}(\R^{d}))}\lesssim
\|w^{\varepsilon}\|_{L^{1}([0,t];
L^{2}(\R^{d}))}+\varepsilon^{-1}\|L^{\varepsilon}\|_{L^{1}([0,t];
L^{2}(\R^{d}))}, \ \ \ \forall t\in [0,T].
\end{equation*}
Thanks to Assumption ~\ref{hyp:V} and Proposition~\ref{prop:global},
we have
\begin{equation}\label{5.9}
\|L^{\varepsilon}\|_{L^{1}([0,t];
L^{2}(\R^{d}))}\lesssim \eps^{3/2}\|\<y\>^{3}u\|_{L^{1}([0,t];
L^{2}(\R^{d}))}\lesssim \varepsilon^{3/2}e^{C_{0}t},
\end{equation}
for any $t\ge 0$. Then the Gronwall inequality yields
\begin{equation*}
\|w^{\varepsilon}\|_{L^{\infty}([0,T]; L^{2}(\R^{d}))}\le
C(T)\sqrt{\varepsilon},
\end{equation*}
where $C(T)$ is independent of $\eps$.

 To establish the
control of $\mathcal{H}$ norm, in view of Lemma~\ref{lem:operators},
we obtain
\begin{equation*}
\(i\varepsilon\partial_{t}+\frac{\varepsilon^{2}}{2}\Delta-V\)(A^{\varepsilon}w^{\varepsilon})
=\sqrt{\varepsilon}\(\nabla V(t, x)-\nabla V(t,
x(t))\)w^{\varepsilon}-A^{\varepsilon}L^{\varepsilon}+A^{\varepsilon}N^{\varepsilon}.
\end{equation*}
 Using Duhamel's formula and scaled Strichartz estimates again, we lead to
\begin{equation}
\label{5.10}
\begin{aligned}
& \|A^{\varepsilon}w^{\varepsilon}\|_{L^{q}L^{r}}\lesssim
\varepsilon^{-1/q}\|A^{\varepsilon}w^{\varepsilon}(t)\|_{L^{2}}
+\varepsilon^{-1-1/q}\|A^{\varepsilon}L^{\varepsilon}\|_{L^{1}L^{2}}\\
&
+\varepsilon^{-1-2/q}\|A^{\varepsilon}N^{\varepsilon}\|_{L^{q^{\prime}}L^{r^{\prime}}}+\varepsilon^{-1-1/q}\|\sqrt{\varepsilon}
\(\nabla V\(t, x\)-\nabla V\(t,
x(t)\)\)w^{\varepsilon}\|_{L^{1}L^{2}}.
\end{aligned}
\end{equation}
Note that we have the pointwise estimate:
\begin{equation*}
|\sqrt{\varepsilon} (\nabla V(t, x)-\nabla V(t,
x(t)))w^{\varepsilon}|\le C\varepsilon|B^{\varepsilon}w^{\varepsilon}|,
\end{equation*}
where $C$ is independent of $t,x$ and $\eps$.
We also have
\begin{equation*}
\|A^{\varepsilon}L^{\varepsilon}\|_{L^{2}(\R^{d})}\lesssim
\varepsilon^{3/2}\(\|\<y\>^{2}u\|_{L^{2}(\R^{d})}+\|\<y\>^{3}\nabla
u\|_{L^{2}(\R^{d})}\).
\end{equation*}
 Having in mind Proposition~\ref{prop:global}, we infer
\begin{equation}\label{5.11}
\|A^{\varepsilon}L^{\varepsilon}\|_{L^{1}([0, t];
L^{2}(\R^{d}))}\lesssim \varepsilon^{3/2}e^{C_{0}t},\quad \forall
t\ge 0.
\end{equation}
 We observe that $A^{\eps}$ acts on gauge
 invariant nonlinearities like a derivative:
\begin{equation*}
A^{\varepsilon}(|x|^{-\gamma}*|\phi|^{2}\phi)=2\RE\(
|x|^{-\gamma}*(\overline{\phi}A^{\varepsilon}\phi)\)\phi+
|x|^{-\gamma}*|\phi|^{2}(A^{\varepsilon}\phi).
\end{equation*}
Then we have
\begin{equation*}
\begin{aligned}
\|A^{\varepsilon}N^{\varepsilon}\|_{L^{q^{\prime}}L^{r^{\prime}}}&\lesssim
\varepsilon^{\alpha_{c}}
\left\lVert |x|^{-\gamma}*
(\overline{\psi^{\varepsilon}}A^{\varepsilon}\psi^{\varepsilon})\psi^{\varepsilon}-
|x|^{-\gamma}*
(\overline{\varphi^{\varepsilon}}A^{\varepsilon}\varphi^{\varepsilon})
\varphi^{\varepsilon}\right\rVert_{L^{q^{\prime}}L^{r^{\prime}}}\\
+\eps^{\alpha_c}&\left\lVert(|x|^{-\gamma}*|\psi^{\varepsilon}|^{2})
A^{\varepsilon}\psi^{\varepsilon}-
(|x|^{-\gamma}*|\varphi^{\varepsilon}|^{2})A^{\varepsilon}
\varphi^{\varepsilon}\right\rVert_{L^{q'}L^{r'}}=:\eps^{\alpha_{c}}(I+II).
\end{aligned}
\end{equation*}
In view of triangle inequality, H\"older inequality and 
Hardy--Littlewood--Sobolev inequality, by similar arguments as in
the proof of Lemma~\ref{lem:localhartree}, we get
\begin{align*}
I&\lesssim
\|A^{\varepsilon}w^{\varepsilon}\|_{L^{q}L^{r}}
\|\psi^{\varepsilon}\|_{L^{\theta}L^{r}}^{2}
+\|w^{\varepsilon}\|_{L^{q}L^{r}}
\|A^{\varepsilon}\varphi^{\varepsilon}\|_{L^{\theta}L^{r}}
(\|\psi^{\varepsilon}\|_{L^{\theta}L^{r}} +
\|\varphi^{\varepsilon}\|_{L^{\theta}L^{r}})\\
&\lesssim
\(\|w^{\varepsilon}\|_{L^{\theta}L^{r}}^{2}
+\|\varphi^{\varepsilon}\|_{L^{\theta}L^{r}}^{2}\)
\|A^{\varepsilon}w^{\varepsilon}\|_{L^{q}L^{r}}
\\
&\quad+\(\|A^{\varepsilon}\varphi^{\varepsilon}\|_{L^{\theta}L^{r}}^{2}
+\|w^{\varepsilon}\|_{L^{\theta}L^{r}}^{2}
+\|\varphi^{\varepsilon}\|_{L^{\theta}L^{r}}^{2}\)\|w^{\varepsilon}\|_{L^{q}L^{r}}.
\end{align*}
The term $II$ satisfies the same estimate.
Applying  Gagliardo--Nirenberg inequality, 
\begin{align*}
\|A^{\varepsilon}\varphi^{\varepsilon}(t)\|_{L^{r}(\R^{d})}&\lesssim
\varepsilon^{-\gamma/8}\|A^{\varepsilon}\varphi^{\varepsilon}\|_{L^{2}(\R^{d})}^{1-\gamma/4}
\|(A^{\varepsilon})^{2}\varphi^{\varepsilon}\|_{L^{2}(\R^{d})}^{\gamma/4}\\
&\lesssim \varepsilon^{-\gamma/8}\|\nabla
u\|_{L^{2}(\R^{d})}^{1-\gamma/4}\|\nabla^{2}u\|_{L^{2}(\R^{d})}^{\gamma/4}\le
C(T)\eps^{-\gamma/8},\quad \forall t\in [0, T],
\end{align*}
which implies that
\begin{equation*}
\|A^{\varepsilon}N^{\varepsilon}\|_{L^{q^{\prime}}L^{r^{\prime}}}\lesssim
\varepsilon^{\alpha_{c}-\gamma/4}\tau^{1-\gamma/4}
(\|A^{\varepsilon}w^{\varepsilon}\|_{L^{q}L^{r}}+\|w^{\varepsilon}\|_{L^{q}L^{r}}).
\end{equation*}
 Then \eqref{5.10} can be estimated as
\begin{equation}\label{5.12}
\begin{aligned}
\|A^{\varepsilon}w^{\varepsilon}\|_{L^{q}L^{r}}&\lesssim\varepsilon^{-1/q}
\|A^{\varepsilon}w^{\varepsilon}(t) \|_{L^{2}}
+\varepsilon^{-1/q}\|B^{\varepsilon}w^{\varepsilon}\|_{L^{1}L^{2}}\\
&
+\varepsilon^{-1-1/q}\|A^{\varepsilon}L^{\varepsilon}\|_{L^{1}L^{2}}
+\tau^{1-\gamma/4}\(\|A^{\varepsilon}w^\varepsilon\|_{L^{q}L^{r}}
+\|w^{\varepsilon}\|_{L^{q}L^{r}}\).
\end{aligned}
\end{equation}
Recalling Lemma ~\ref{lem:operators}, we have
\begin{equation*}
\(i\varepsilon\partial_{t}+\frac{\varepsilon^{2}}{2}\Delta
-V\)(B^{\varepsilon}w^{\varepsilon})=\varepsilon
A^{\varepsilon}w^{\varepsilon}-B^{\varepsilon}L^{\varepsilon}+B^{\varepsilon}N^{\varepsilon}.
\end{equation*}
Proceeding like above, we come up with
\begin{equation}\label{5.16}
\begin{aligned}
\|B^{\varepsilon}w^{\varepsilon}\|_{L^{q}L^{r}}&\lesssim
\varepsilon^{-1/q}\|B^{\varepsilon}w^{\varepsilon}(t)\|_{L^{2}}
+\varepsilon^{-1/q}\|A^{\varepsilon}w^{\varepsilon}\|_{L^{1}L^{2}}\\
&\quad
+\varepsilon^{-1-1/q}\|B^{\varepsilon}L^{\varepsilon}\|_{L^{1}L^{2}}
+\tau^{1-\gamma/4}\(\|B^{\varepsilon}w^{\varepsilon}\|_{L^{q}L^{r}}
+\|w^{\varepsilon}\|_{L^{q}L^{r}}\).
\end{aligned}
\end{equation}
Summing over  \eqref{5.8}, \eqref{5.12} and \eqref{5.16}, we have
\begin{align*}
\sum_{P^{\eps}\in \{{\rm Id}, A^{\eps},
B^{\eps}\}}&\|P^{\eps}w^{\varepsilon}\|_{L^{q}L^{r}}\lesssim
\varepsilon^{-1/q}\sum_{P^{\eps}\in \{{\rm Id}, A^{\eps},
B^{\eps}\}}\(\|P^{\eps}w^{\varepsilon}(t)\|_{L^{2}}
+\|P^{\varepsilon}w^{\varepsilon}\|_{L^{1}L^{2}}\)
\\
+\varepsilon^{-1-1/q}&\sum_{P^{\eps}\in \{{\rm Id}, A^{\eps},
B^{\eps}\}}\|P^{\eps}L^{\varepsilon}\|_{L^{1}L^{2}}
+\tau^{1-\gamma/4}\sum_{P^{\eps}\in \{{\rm Id}, A^{\eps},
B^{\eps}\}}\|P^{\eps}w^{\varepsilon}\|_{L^{q}L^{r}}.
\end{align*}
Take $\tau$ sufficiently small such that the last term of the right
hand side in the above inequality can be absorbed by the left hand
side. Using scaled Strichartz estimates again, resuming the above
computations, for any fixed $T>0$ and any $t\in [0, T]$, 
\begin{align*}
\sum_{P^{\eps}\in \{{\rm Id}, A^{\eps},
B^{\eps}\}}\|P^{\eps}w^{\varepsilon}\|_{L^{\infty}([0, t];
L^{2}(\R^{d}))}& \lesssim \sum_{P^{\eps}\in \{{\rm Id}, A^{\eps},
B^{\eps}\}}\|P^{\eps}w^{\varepsilon}\|_{L^{1}([0, t];
L^{2}(\R^{d}))}\\
&\quad +\varepsilon^{-1}\sum_{P^{\eps}\in \{{\rm Id}, A^{\eps},
B^{\eps}\}}\|P^{\eps}L^{\varepsilon}\|_{L^{1}([0, t]; L^{2}(\R^{d}))}.
\end{align*}
We end up with
\begin{equation}
\|w^{\varepsilon}\|_{L^{\infty}([0, t]; \mathcal{H})}\lesssim
\|w^{\varepsilon}\|_{L^{1}([0, t]; \mathcal{H})}+\sqrt{\varepsilon},
\quad \forall t\in [0, T],
\end{equation}
and Proposition~\ref{prop:bounded} follows from Gronwall lemma.
\end{proof}

\section{Large time approximation}\label{sec:largetime}

 In this section, we improve the time of validity of the error
 estimate proven in \S\ref{sec:bounded}, in two cases: the
linearizable case $\alpha>\alpha_{c}$, and the non-linearizable case
$\alpha=\alpha_{c}$, thus proving Proposition~\ref{prop:subhartree}
and Theorem~\ref{thm:criticalhartree}.

\subsection{Proof of Proposition ~\ref{prop:subhartree}} Set
$w^{\varepsilon}=\psi^{\varepsilon}-\varphi^{\varepsilon}_{{\rm
lin}}$: it satisfies
\begin{equation}\label{6.1}
   i\varepsilon\partial_{t}w^{\varepsilon}+\frac{\varepsilon^{2}}{2}\Delta
u= V w^{\varepsilon}-L^{\varepsilon}+N^{\varepsilon}\quad;\quad
 w^{\eps}|_{t=0}=0,
\end{equation}
where
\begin{equation*}
L^{\varepsilon}=(V(t, x)-T_{2}(t,
x,x(t)))\varphi^{\varepsilon}_{{\rm lin}}\quad; \quad
N^{\varepsilon}=\lambda\varepsilon^{\alpha}
|x|^{-\gamma}*|\varphi^{\varepsilon}_{{\rm
lin}}+w^{\varepsilon}|^{2} (\varphi^{\varepsilon}_{{\rm
lin}}+w^{\varepsilon}).
\end{equation*}
Using scaled Strichartz estimates, we have, for $t>0$,
\begin{align*}
\|w^{\varepsilon}\|_{L^{q}_t L^{r}}&\lesssim
 \varepsilon^{-1-1/q}\|L^{\varepsilon}\|_{L^{1}_t L^{2}}
+\varepsilon^{-1-2/q}\|N^{\varepsilon}\|_{L^{q'}_t L^{r'}}\\
&\lesssim\varepsilon^{-1-1/q}\|L^{\varepsilon}\|_{L^{1}_t
L^{2}} +\varepsilon^{\alpha-1-2/q}
\|\varphi_{{\rm lin}}^{\varepsilon}\|_{L^{\theta}_t L^{r}}^{2}\|\varphi_{{\rm
lin}}^{\varepsilon}\|_{L^{q}_tL^{r}}\\
&\quad 
+\varepsilon^{\alpha-1-2/q}\(\|w^{\varepsilon}\|_{L^{\theta}_t
L^{r}}^{2}+\|\varphi_{{\rm
lin}}^{\varepsilon}\|_{L^{\theta}_t
L^{r}}^{2}\)\|w^{\varepsilon}\|_{L^{q}_t L^{r}},
\end{align*}
where $L^a_tL^b$ stands for $L^a([0,t];L^b(\R^d))$. In view of
Proposition~\ref{prop:global}, 
Gagliardo--Nirenberg inequality yields
\begin{equation*}
\|\varphi_{{\rm lin}}^{\varepsilon}(t)\|_{L^{r}(\R^{d})}\lesssim
\varepsilon^{-\gamma/8}\|\varphi_{{\rm
lin}}^{\varepsilon}\|_{L^{2}(\R^{d})}^{1-\gamma/4}
\|A^{\varepsilon}\varphi_{{\rm
lin}}^{\varepsilon}\|_{L^{2}(\R^{d})}^{\gamma/4} \le
C\varepsilon^{-\gamma/8}e^{C_0 t},\quad \forall t\ge  0,
\end{equation*}
where $C$ and $C_0$ are independent of $\eps$ and $t$. 
We use a
bootstrap argument, relying upon the estimate
\begin{equation}\label{eq:solong}
\|w^{\varepsilon}(t)\|_{L^{r}(\R^{d})}\le
\varepsilon^{-\gamma/8}e^{C_0 t},
\end{equation}
for $t\in [0, T^\eps]$, where $T^\eps$ may (and will) depend on
$\eps$. Since $w^\eps$ is expected to be small compared to
$\varphi_{\rm lin}^\eps$, such an estimate looks sensible. We come up
with
\begin{equation*}
  \|w^{\varepsilon}\|_{L^{q}_t L^{r}}
\le C\varepsilon^{-1-1/q}\|L^{\varepsilon}\|_{L^{1}_t
L^{2}} +C\varepsilon^{\alpha-\alpha_c}e^{2C_0t}\|\varphi_{{\rm
lin}}^{\varepsilon}\|_{L^{q}_tL^{r}}
+C\varepsilon^{\alpha-\alpha_c}e^{2C_0t}\|w^{\varepsilon}\|_{L^{q}_tL^{r}}.
\end{equation*}
Now assume $T^\eps$ is chosen so that 
\begin{equation*}
C\varepsilon^{\alpha-\alpha_{c}}e^{2C_0T^\eps}\le \frac{1}{2}.
\end{equation*}
Then we have 
\begin{equation*}
  \|w^{\varepsilon}\|_{L^{q}_t L^{r}}
\le 2C\varepsilon^{-1-1/q}\|L^{\varepsilon}\|_{L^{1}_t
L^{2}} +2C\varepsilon^{\alpha-\alpha_c}e^{2C_0t}\|\varphi_{{\rm
lin}}^{\varepsilon}\|_{L^{q}_tL^{r}}.
\end{equation*}
Using scaled Strichartz
estimates again, we get
\begin{align*}
\|w^{\varepsilon}\|_{L^{\infty}_tL^{2}}&\lesssim
\varepsilon^{-1}\|L^{\varepsilon}\|_{L^{1}_tL^{2}}
+\varepsilon^{\alpha-1-1/q}e^{2C_0 t}\|\varphi_{{\rm
lin}}^{\varepsilon}\|_{L^{q}_tL^{r}}\\
&\le
C\(\sqrt{\varepsilon}e^{C_{0}t}+\varepsilon^{\alpha-\alpha_{c}}e^{3C_{0}t}\).
\end{align*}
Apply $A^{\varepsilon}$ and $B^{\varepsilon}$ respectively to 
Equation~\eqref{6.1}, then by the same arguments as in the proof
of Proposition~\ref{prop:bounded}, we get
\begin{align*}
& \sum_{P^{\eps}\in \{A^{\eps},
B^{\eps}\}}\|P^{\varepsilon}w^{\varepsilon}\|_{L^{q}_t
L^{r}}\lesssim \varepsilon^{-1/q}\sum_{P^{\eps}\in \{ A^{\eps},
B^{\eps}\}}\|P^{\varepsilon}w^{\varepsilon}\|_{L^{1}_t
L^{2}}+ \varepsilon^{\alpha-\alpha_{c}}e^{3C_{0}t}\\
&\quad 
+\varepsilon^{-1-1/q} \sum_{P^{\eps}\in \{A^{\eps},
B^{\eps}\}}\|P^{\varepsilon}L^{\varepsilon}\|_{L^{1}_t
L^{2}}+ \varepsilon^{\alpha-\alpha_{c}}e^{2C_{0}t}
\sum_{P^{\eps}\in \{{\rm Id}, A^{\eps},
B^{\eps}\}}\|P^{\eps}w^{\varepsilon}\|_{L^{q}_tL^{r}}.
\end{align*}
As long as \eqref{eq:solong} holds, we can use the same absorption
argument as above to treat the last term. By
using scaled Strichartz again, we obtain
\begin{equation*}
\|w^{\varepsilon}\|_{L^{\infty}([0, t]; \mathcal{H})}\lesssim
\|w^{\varepsilon}\|_{L^1([0, t];
\mathcal{H})}+\varepsilon^{\alpha-\alpha_{c}} e^{3C_{0}t}
+\varepsilon^{-1}\|L^{\varepsilon}\|_{L^{1}([0, t];
\mathcal{H})}.
\end{equation*}
Since
\begin{equation*}
\varepsilon^{-1}\|L^{\varepsilon}\|_{L^{1}([0, t];
\mathcal{H})}\le C\sqrt{\varepsilon}e^{Ct},
\end{equation*}
we end up with
\begin{equation*}
\|w^{\varepsilon}\|_{L^{\infty}([0, t]; \mathcal{H})}\le C\(
\|w^{\varepsilon}\|_{L^1([0, t];
\mathcal{H})}+\varepsilon^{\kappa}e^{C_{1}t}\),
\end{equation*}
where $\kappa=\min\{\frac{1}{2}, \alpha-\alpha_{c}\}$ and $C_{1}$ is
independent of $\eps$ and $t$. Gronwall lemma yields
\begin{equation*}
\|w^\eps(t)\|_{\mathcal{H}}\le
C\varepsilon^{\kappa}e^{C_{1}t}
\end{equation*}
so long as \eqref{eq:solong} holds.  
Gagliardo--Nirenberg inequality yields
\begin{equation*}
\|w^{\varepsilon}(t)\|_{L^{r}(\R^{d})}\lesssim
\varepsilon^{-\gamma/8}\|w^{\varepsilon}\|_{L^{2}(\R^{d})}^{1-\gamma/4}
\|A^{\varepsilon}w^{\varepsilon}\|_{L^{2}(\R^{d})}^{\gamma/4}\le
C\varepsilon^{\kappa-\gamma/8}e^{C_{1}t}.
\end{equation*}
Proposition ~\ref{prop:subhartree} then follows, by choosing $T^\eps
=C_2\ln\frac{1}{\eps}$ with $C_2>0$ sufficiently small and independent 
of $\eps$, and $\eps\in ]0,\eps_0]$ for $\eps_0>0$ sufficiently
small.

\subsection{Proof of Theorem ~\ref{thm:criticalhartree}}
We know explain how to upgrade Proposition~\ref{prop:bounded} to
Theorem ~\ref{thm:criticalhartree}, by examining the large time
behavior of the quantities involved in the proof. 
By Proposition~\ref{prop:global}, we know that $u\in C(\R_{+},
\Sigma^{k})$ provided that the initial datum $a\in
\Sigma^{k}$. Recalling Step 1 in Lemma~\ref{lem:exphartree}, we
showed that for any $t\ge 0$ and $\tau>0$ sufficiently small,
\begin{equation*}
\|u\|_{L^{q}([t, t+\tau]; L^{r}(\R^{d}))}\le
\frac{3}{2}C\|u(0,\cdot)\|_{L^{2}(\R^{d})}=\frac{3}{2}C\|a\|_{L^{2}(\R^{d})}.
\end{equation*}
where $C$ is independent of $t,\tau$. For any $t\ge  0$, fixed
$\tau\in ]0, 1[$, split $[t, t+1]$ into finitely many pieces with
length at most $\tau$, then we obtain
\begin{equation*}
\|\varphi^{\varepsilon}\|_{L^{q}([t, t+1];
L^{r}(\R^{d}))}=\eps^{-\gamma/8}\|u\|_{L^{q}([t, t+1];
L^{r}(\R^{d}))}\le C\eps^{-\gamma/8}\|a\|_{L^{2}(\R^{d})}, \quad
\forall t\ge  0.
\end{equation*}
Using Proposition~\ref{prop:bounded} and Gagliardo--Nirenberg
inequality, we know that there exists $\eps_0>0$ such that 
\begin{equation}\label{6.5}
\|w^{\varepsilon}\|_{L^{q}([t, t+1]; L^{r}(\R^{d}))}\le
\eps^{-\gamma/8}\|a\|_{L^{2}(\R^{d})},
\end{equation}
for $t\in [0,1]$. Suppose that \eqref{6.5} holds for $t\in [0,T^\eps]$
($T^\eps\ge 1$), and let $t,\tau>0$ with $t+\tau\le T^\eps$. Denote by
$L^aL^b$ the space $L^a([t,t+\tau];L^b(\R^d))$. 
Scaled Strichartz estimates, H\"older inequality and
Hardy--Littlewood--Sobolev inequality yield
\begin{align*}
\|w^{\varepsilon}\|_{L^{q}L^{r}}&\lesssim
\eps^{-1/q}\|w^\eps(t)\|_{L^2}+ 
\varepsilon^{-1-1/q}\|L^{\varepsilon}\|_{L^{1}L^{2}}\\
&\quad
+\varepsilon^{\alpha_{c}-1-2/q}\(\|\varphi^{\varepsilon}\|_{L^{\theta}L^{r}}^{2}+
\|w^{\varepsilon}\|_{L^{\theta}L^{r}}^{2}\)\|w^{\varepsilon}\|_{L^{q}L^{r}}\\
&\le
C\(\varepsilon^{-1/q}\|w^{\varepsilon}(t)\|_{L^{2}}+
\varepsilon^{-1-1/q}\|L^{\varepsilon}\|_{L^{1}L^{2}}
+\tau^{1-\gamma/2}\|w^{\varepsilon}\|_{L^{q}L^{r}}\),
\end{align*}
where $C$ is independent of $\eps$, $t$ , $\tau$. Choose $\tau\in ]0,
1]$ sufficiently small 
such that the last term on the above right hand side can
be absorbed by the left hand side. Using scaled Strichartz estimates
again and resuming the previous computations, we end up with
\begin{equation*}
\|w^{\varepsilon}\|_{L^{\infty}([0, t]; L^{2}(\R^{d}))} \le
C\int_{0}^{t}\|w^{\varepsilon}\|_{L^{\infty}([0, s];
L^{2}(\R^{d}))}ds+C\varepsilon^{-1}\int_{0}^{t}\|L^{\varepsilon}(s)\|_{L^{2}}ds.
\end{equation*}
The last term is controlled thanks to Proposition~\ref{prop:global},
and   Gronwall lemma yields
\begin{equation}\label{6.7}
\|w^{\varepsilon}\|_{L^{\infty}([0, t]; L^{2}(\R^{d}))}\le
C\sqrt{\varepsilon}e^{C_{1}t}.
\end{equation}
Resuming Strichartz inequalities, and splitting the interval $[t,t+1]$
into finitely many intervals of length $\tau$, the above computation
yields: 
\begin{align*}
  \|w^\eps\|_{L^q([t,t+1];L^r)}&\lesssim \eps^{-1/q}\|w^\eps(t)\|_{L^2}
  + \eps^{-1-1/q}\|L^\eps\|_{L^1([t,t+1];L^2)}\lesssim \eps^{-\g/8}
  \sqrt\eps e^{C_1 t}.  
\end{align*}
Setting $T^\eps
=C_2\ln\frac{1}{\eps}$ with $C_2>0$ sufficiently small and independent 
of $\eps$, and $\eps\in ]0,\eps_0]$ for $\eps_0>0$ sufficiently
small, we see that \eqref{6.5} holds for all $t\in [0,T^\eps]$; the
first part of Theorem~\ref{thm:criticalhartree} follows.

As in the proof of Proposition ~\ref{prop:bounded}, we also have
\begin{align*}
& \sum_{P^{\eps}\in \{{\rm Id}, A^{\eps},
B^{\eps}\}}\|P^{\eps}w^{\eps}\|_{L^{q}L^{r}}\lesssim
\eps^{-1/q}\sum_{P^{\eps}\in \{{\rm Id}, A^{\eps},
B^{\eps}\}}\(\|P^{\eps}w^{\eps}(t)\|_{L^{2}}+\|P^{\eps}w^{\eps}\|_{L^{1}L^{2}}\)\\
& \quad +\eps^{1-1/q}\sum_{P^{\eps}\in \{{\rm Id}, A^{\eps},
B^{\eps}\}}\|P^{\eps}L^{\eps}\|_{L^{1}L^{2}}
+\eps^{\alpha_{c}-1-2/q}\|w^{\eps}\|_{L^{q}L^{r}}\sum_{P^{\eps}\in
\{A^{\eps},
B^{\eps}\}}\|P^{\eps}\varphi^{\eps}\|_{L^{\theta}L^{r}}^{2}\\
&\quad
+\eps^{\alpha_{c}-1-2/q}\(\|w^{\eps}\|_{L^{\theta}L^{r}}^{2}+
\|\varphi^{\eps}\|_{L^{\theta}L^{r}}^{2}\)
\sum_{P^{\eps}\in \{{\rm Id}, A^{\eps},
B^{\eps}\}}\|P^{\eps}w^{\eps}\|_{L^{q}L^{r}}.
\end{align*}
Choosing $\tau$ sufficiently small, the last
term on the right hand side of the above inequality can be absorbed
by the left hand side. Notice that so long as \eqref{6.5} holds,
\begin{equation*}
\eps^{\alpha_{c}-1-\frac{2}{q}}\|w^{\eps}\|_{L^{q}L^{r}}
\(\|A^{\eps}\varphi^{\eps}\|_{L^{\theta}L^{r}}^{2}
+\|B^{\eps}\varphi^{\eps}\|_{L^{\theta}L^{r}}^{2}\)\lesssim
\eps^{-\gamma/8}e^{2C_{0}t}.
\end{equation*}
Using scaled Strichartz estimates again, mimicking the above
computations, and using Proposition~\ref{prop:global}, Gronwall lemma
yields
\begin{equation*}
\|w^{\eps}\|_{L^{\infty}([0, t]; \mathcal{H})}\le
C\sqrt{\eps}e^{C_{1}t},
\end{equation*}
which completes the proof of Theorem~\ref{thm:criticalhartree}.

\section{Nonlinear superposition}\label{sec:superposition}
In this section, we consider the nonlinear superposition
for the critical case $\alpha=\alpha_{c}$. The arguments
to prove Theorem~\ref{thm:superpositionbounded} and
Theorem~\ref{thm:superpositionlargetime} are quite similar to the
proof of 
Proposition~\ref{prop:bounded} and Theorem~\ref{thm:criticalhartree},
respectively. The important aspect to be understood is the
interaction between the two profiles $\varphi^{\eps}_{1}$ and
$\varphi^{\eps}_{2}$. The error $
w^{\varepsilon}=\psi^{\varepsilon}-\varphi_{1}^{\varepsilon}-
\varphi_{2}^\varepsilon$ satisfies
\begin{equation*}
i\varepsilon\partial_{t}w^{\varepsilon}+\frac{\varepsilon^{2}}{2}\Delta
w^{\varepsilon}=
Vw^{\varepsilon}-L^{\varepsilon}+N^{\varepsilon}\quad ; \quad
w^{\eps}|_{t=0}=0,
\end{equation*}
where
\begin{align*}
L^{\varepsilon}&=(V(t,x)-T_{2}(t, x,
x(t)))(\varphi_{1}^{\varepsilon}+\varphi^{\varepsilon}_{2}),\\
N^{\varepsilon}&=\lambda\varepsilon^{\alpha_{c}}\(F\(w^{\varepsilon}+
\varphi_{1}^{\varepsilon}     
+\varphi_{2}^{\varepsilon}\)-
F\(\varphi_{1}^{\varepsilon}\)-F\(\varphi_{2}^{\varepsilon}\)\), \quad
F(\psi):=  
\(|x|^{-\gamma}*|\psi|^{2}\)\psi.
\end{align*}
As in \cite{CaFe11}, decompose $N^{\varepsilon}$ into two parts: a 
semilinear term
\begin{equation*}
N_{S}^{\varepsilon}=\lambda\varepsilon^{\alpha_{c}}\(F\(w^{\varepsilon}+
\varphi_{1}^{\varepsilon} +\varphi_{2}^{\varepsilon}\)-F\(
\varphi_{1}^{\varepsilon} +\varphi_{2}^{\varepsilon}\)\),
\end{equation*}
and an interaction (source) term
\begin{equation*}
N_{I}^{\varepsilon}=\lambda\varepsilon^{\alpha_{c}}
\( F\(\varphi_{1}^{\varepsilon} +\varphi_{2}^{\varepsilon}\) -
F\(\varphi_{1}^{\varepsilon}\)-F\(\varphi_{2}^{\varepsilon}\)\).
\end{equation*}
As noticed in \cite{CaFe11}, the term $N_{S}^{\eps}$ can be treated as
in the case of a single wave packet: the
key point is to estimate
\begin{equation*}
\varepsilon^{\alpha_{c}-1}\|N_{I}^{\varepsilon}\|_{L^{1}([0, t];
\Sigma_\eps)},
\end{equation*}
since $N_I^\eps$ plays the role of a source term. 
This is the only new term to control to infer
Theorem~\ref{thm:superpositionbounded} and 
Theorem~\ref{thm:superpositionlargetime}
from 
Proposition~\ref{prop:bounded} and Theorem~\ref{thm:criticalhartree}. 
\begin{lemma}\label{lem:7.1}
Let $0<\gamma<\min (2, d)$, $T\ge 0$ and $0<\si<\frac{1}{2}$.
Denote
\begin{equation}\label{7.2}
I^{\eps}(T)=\left\{t\in [0, T];\quad |x_{1}(t)-x_{2}(t)|\le \eps^{\si}\right\}.
\end{equation}
Then for any $k\in \N$ with $k>\g$, 
\begin{equation*}
\eps^{-1}\|N_{I}^{\eps}\|_{L^{1}([0, T]; \Sigma_\eps)}\lesssim
(M_{k+2}(T))^{3}\(T\eps^{\gamma(1/2-\si)}+|I^\eps(T)|
\)e^{CT},
\end{equation*}
where
$M_{k}(T)=\sup\{\|u_{j}\|_{L^{\infty}([0,
T];\Sigma^k)};\  j\in \{1, 2\}\}.$ \\
In particular, if $T>0$ is independent of $\eps$, then 
\begin{equation*}
\eps^{-1}\|N_{I}^{\eps}\|_{L^{1}([0, T]; \Sigma_\eps)}\lesssim
(M_{k+2}(T))^{3}\(\eps^{\gamma(1/2-\si)}+\eps^\si\).
\end{equation*}
\end{lemma}
\begin{proof}
We compute
\begin{equation}\label{eq:ni}
\begin{aligned}
N_{I}^{\eps}&=\eps^{\alpha_{c}}\(\(|x|^{-\gamma}* 
|\varphi^{\eps}_{1}|^{2}\)\varphi^{\eps}_{2} +
\(|x|^{-\gamma}*|\varphi^{\eps}_{2}|^{2}\) \varphi^{\eps}_{1}\)\\
&\quad
+2\eps^{\alpha_{c}}\(|x|^{-\gamma}*\(\RE
\(\varphi_{1}^{\eps}\overline{\varphi_{2}^{\eps}}\)\)\)
\(\varphi_{1}^{\eps}+\varphi_2^\eps\).
\end{aligned}
\end{equation}
On the complement of $I^\eps(T)$, we will use Peetre inequality: for
$\eta\in\R^{d}$, 
\begin{equation*}
\sup_{x\in\R^{d}}\(\<x\>^{-1}\<x-\eta\>^{-1}\)\lesssim\frac{1}{\<\eta\>}\lesssim
\frac{1}{|\eta|}.
\end{equation*}
Denote $\eta^{\eps}=\frac{x_{1}(t)-x_{2}(t)}{\sqrt{\eps}}$. For the
last term in \eqref{eq:ni},  we have, for $j\in \{1,2\}$,
\begin{align*}
& \eps^{\alpha_{c}-1}\int_{[0, T]\setminus
I^{\eps}(T)}\left\|\(|x|^{-\gamma}*|\varphi_{1}^{\eps}(t,
x)\varphi_{2}^{\eps}(t, x)|\)\varphi_{j}^{\eps}(t,
x)\right\|_{L^{2}(\R^{d})}dt\\
 &=\int_{[0, T]\setminus
I^{\eps}(T)}\left\|\(|x|^{-\gamma}*|u_{1}(t,
x)u_{2}(t, x-\eta^{\eps})|\)u_{j}(t, x)\right\|_{L^{2}(\R^{d})}dt\\
&\lesssim\left\||x|^{-\gamma}*|\<x\>^{k}u_{1}(t,
x)\<x-\eta^{\eps}\>^{k}u_{2}(t, x-\eta^{\eps})|
\right\|_{L^{\infty}L^{2d/\g}}
 \|u_{j}\|_{L^{\infty}L^{2d/(d-\g)}} \\
&\quad \times \int_{[0, T]
 \setminus
I^{\eps}(T)}\frac{dt}{|\eta^{\eps}(t)|^{k}}.
\end{align*}
In view of Hardy--Littlewood--Sobolev inequality and Sobolev
embedding, we have:
\begin{align*}
& \quad \left\||x|^{-\gamma}*|\<x\>^{k}u_{1}(t,
x)\<x-\eta^{\eps}\>^{k}u_{2}(t,
x-\eta^{\eps}) |\right\|_{L^{\infty}L^{2d/\g}}\\
&\lesssim \|\<x\>^{k}u_{1}(x)\|_{L^{\infty}L^{4d/(2d-\gamma)}}
\|\<x-\eta^{\eps}\>^{k}u_{2}(x-\eta^{\eps})\|_{L^{\infty}L^{4d/(2d-\gamma)}}\\
&\lesssim \|\<x\>^{k}u_{1}(x)\|_{L^{\infty}H^{1}}
\|\<x-\eta^{\eps}\>^{k}u_{2}(x-\eta^{\eps})\|_{L^{\infty}H^{1}}\lesssim
(M_{k+1}(T))^{2}.
\end{align*}
On the other hand, we have
\begin{equation*}
\int_{[0, T]\setminus I^{\eps}(T)}\frac{dt}{|\eta^{\eps}(t)|^{k}}
=\int_{[0, T]\setminus
I^{\eps}(T)}\frac{\eps^{k/2}}{|x_{1}(t)-x_{2}(t)|^{k}}dt
\le\eps^{k(1/2-\si)}T.
\end{equation*}
It follows that
\begin{equation*}
\eps^{\alpha_{c}-1}\int_{[0, T]\setminus
I^{\eps}(T)}\left\|(|x|^{-\gamma}*|\varphi_{1}^{\eps}
\varphi_{2}^{\eps}|)\varphi_{j}^{\eps}(t,
x)\right\|_{L^{2}(\R^{d})}dt\lesssim (M_{k+1}(T))^{3}\eps^{k(1/2-\si)}T.
\end{equation*}
The first two terms in \eqref{eq:ni} are of the same form, so we
consider the first one only:
\begin{align*}
&\quad\eps^{\alpha_{c}-1}\int_{[0, T]\setminus
I^{\eps}(T)}\left\|\(|x|^{-\gamma}*|\varphi_{1}^{\eps}|^{2}\)
\varphi_{2}^{\eps}\right\|_{L^{2}(\R^{d})}dt\\
&=\int_{[0, T]\setminus
I^{\eps}(T)}\left\|\(\int_{\R^{d}}|y|^{-\gamma}|u_{1}(t,
x-y-\eta^{\eps})|^{2}dy\)u_{2}(t,x)\right\|_{L^{2}(\R^{d})}dt.
\end{align*}
In view of Peetre inequality, using similar arguments as above, we
have
\begin{align*}
&\quad \int_{[0, T]\setminus
I^{\eps}(T)}\left\|\int_{|y+\eta^{\eps}|>
|\eta^{\eps}|}|y|^{-\gamma}|u_{1}(t, x-y-\eta^{\eps})|^{2}dy\times
u_{2}(t,
x)\right\|_{L^{2}(\R^{d})}dt\\
\lesssim &\(\int_{[0, T]\setminus I^{\eps}(T)}\sup_{{x\in\R^{d}}\atop
{|y+\eta^{\eps}|>
|\eta^{\eps}|}}\<x-y-\eta^{\eps}\>^{-k}\<x\>^{-k}dt\)\|\<x\>^{k}u_{2}(t,
x)\|_{L^{\infty}L^{2d/(d-\g)}}\\
&  \times\left\|\int |y|^{-\gamma}|u_{1}(t,
x-y-\eta^{\eps})|^{2}\<x-y-\eta^{\eps}\>^{k}dy\right\|_{L^{\infty}L^{2d/\g}}\\
\lesssim &\int_{[0, T]\setminus
I^{\eps}(T)}\frac{dt}{|\eta^{\eps}|^{k}}
\left\|\<x\>^{k}u_{2}(t,x)\right\|_{L^{\infty}L^{2d/(d-\g)}}
\left\|\<x\>^{k/2} u_{1}(t,x)\right\|_{L^{\infty}L^{4d/(2d-\gamma)}}^2
\\
\lesssim &(M_{k+1}(T))^{3}\eps^{k(1/2-\si)}T.
\end{align*}
We observe that
\begin{equation*}
\{y\in\R^{d}: |y+\eta^{\eps}|\le |\eta^{\eps}|\}\subset \{y\in
\R^{d}: |y|\le 2|\eta^{\eps}|\},
\end{equation*}
and, for $k>0$, 
$\<x-\eta^{\eps}\>^{k}\lesssim \<x-y-\eta^{\eps}\>^{k}+|y|^{k}$.
Then we have, for all $x\in \R^d$,
 \begin{align*}
& \quad \int_{\{y:|y+\eta^{\eps}|\le
|\eta^{\eps}|\}}|y|^{-\gamma}|u_{1}(t, x-y-\eta^{\eps})|^{2}
\<x-\eta^{\eps}\>^{k}dy\\
&\lesssim \int_{\R^{d}} |y|^{-\gamma}|u_{1}(t,
x-y-\eta^{\eps})|^{2} \<x-y-\eta^{\eps}\>^{k}dy\\
& \quad +\int_{\{y:|y|\le 2|\eta^{\eps}|\}}
|y|^{k-\gamma}|u_{1}(t, x-y-\eta^{\eps})|^{2}dy\\
&\lesssim
 \int_{\R^{d}} |y|^{-\gamma}|u_{1}(t,
x-y-\eta^{\eps})|^{2} \<x-y-\eta^{\eps}\>^{k}dy
+|\eta^{\eps}|^{k-\gamma}\|u_{1}(t)\|_{L^{2}}^{2},
\end{align*}
since $k>\g$. 
It follows that
\begin{align*}
& \quad \int_{[0, T]\setminus
I^{\eps}(T)}\left\|\(\int_{|y+\eta^{\eps}|\le
|\eta^{\eps}|}|y|^{-\gamma}|u_{1}(t, x-y-\eta^{\eps})|^{2}dy\)
u_{2}(t,
x)\right\|_{L^{2}(\R^{d})}dt\\
\lesssim &\(\int_{[0, T]\setminus I^{\eps}(T)}\sup_{{x\in\R^{d}}\atop
{|y+\eta^{\eps}|>
|\eta^{\eps}|}}\<x-\eta^{\eps}\>^{-k}\<x\>^{-k}dt\)\|\<x\>^{k}u_{2}(t,
x)\|_{L^{\infty}L^{2d/(d-\g)}}\\
&  \times\left\|\int |y|^{-\gamma}|u_{1}(t,
x-y-\eta^{\eps})|^{2}\<x-\eta^{\eps}\>^{k}dy\right\|_{L^{\infty}L^{2d/\g}}\\
& \lesssim  (M_{k+1}(T))^{3}\int_{[0, T]\setminus
I^{\eps}(T)}{|\eta^{\eps}|^{-k}}(1+|\eta^{\eps}|^{k-\gamma})dt\\
& \lesssim
(M_{k+1}(T))^{3}\(\eps^{k(1/2-\si)}+\eps^{\gamma(1/2-\si)}\)T\lesssim
(M_{k+1}(T))^{3}\eps^{\gamma(1/2-\si)}T,
\end{align*}
since $k>\g$. In $I^\eps(T)$, 
H\"older inequality, Hardy--Littlewood--Sobolev inequality and Sobolev
embedding yield similarly
\begin{equation*}
\eps^{\alpha_{c}-1}\int_{I^{\eps}(T)}\|N_{I}^{\eps}\|_{L^{2}(\R^{d})}dt\lesssim
(M_{1}(T))^{3}|I^{\eps}(T)|.
\end{equation*}
Note that
\begin{align*}
\|\sqrt{\eps}\nabla\varphi_{j}\|_{L^{2}(\R^{d})}&\lesssim\sqrt{\eps}\|\nabla
u_{j}\|_{L^{2}(\R^{d})}+|\xi_{j}|\|u_{j}\|_{L^{2}(\R^{d})},\\
\|x\varphi_{j}\|_{L^{2}(\R^{d})}&\lesssim \sqrt{\eps}\|y
u_{j}\|_{L^{2}(\R^{d})}+|x_{j}|\|u_{j}\|_{L^{2}(\R^{d})}.
\end{align*}
The $\Sigma_\eps$ estimate of $N_I^\eps$ then follows easily, in view
of Lemma ~\ref{lem:traj}.

When $T>0$ does not depend on $\eps$, we simply invoke 
 Lemma 6.2 in \cite {CaFe11}:
\begin{equation*}
|I^{\eps}(T)|=\O\(\eps^{\si}\),
\end{equation*}
and the proof of the lemma is complete.
\end{proof}

\subsection{Proof of Theorem ~\ref{thm:superpositionbounded}} Due to
the lack of natural rescaling for two 
wave packets, we  use a bootstrap argument even on finite time
intervals; the operators $A^\eps$ and $B^\eps$ used in
\S\ref{sec:bounded} are helpful analytically because they have
a precise geometrical meaning in terms of one wave packet, and this
meaning is lost in the case of two wave packets. Since for 
$j=1,2$,
\begin{equation*}
\|\varphi_{j}^{\eps}(t)\|_{L^{r}(\R^{d})}\le C(T)\eps^{-\gamma/8},
\quad \forall t\in [0, T],
\end{equation*}
the bootstrap argument goes as follows: so long as
\begin{equation}\label{7.4}
\|w^{\eps}(t)\|_{L^{r}(\R^{d})}\le C(T)\eps^{-\gamma/8},
\end{equation}
we estimate the error $w^\eps$ with a rather precise rate.
Resuming the
computations of Proposition ~\ref{prop:bounded} and using the above
bootstrap argument, we have
\begin{equation*}
\|w^{\eps}\|_{L^{\infty}([0, T]; L^{2}(\R^{d}))}\lesssim
\frac{1}{\eps}\|L^{\eps}\|_{L^{1}([0, T];
L^{2}(\R^{d}))}+\frac{1}{\eps}\|N_{I}^{\eps}\|_{L^{1}([0, T];
L^{2}(\R^{d}))}.
\end{equation*}
Similarly, applying $\sqrt{\eps}\nabla$ and $x$ to the equation,
resuming an analogue computation as $A^{\eps}$ and $B^{\eps}$, we
get
\begin{align*}
\|w^{\eps}\|_{L^{\infty}([0, T]; \Sigma_{\eps})}&\lesssim
\frac{1}{\eps}\|L^{\eps}\|_{L^{1}([0, T];
\Sigma_{\eps})}+\frac{1}{\eps}\|N_{I}^{\eps}\|_{L^{1}([0, T];
\Sigma_{\eps})}\\
&\lesssim\sqrt{\eps}+\eps^{\si}+\eps^{\gamma(1/2-\si)}\lesssim
\eps^{\si}+\eps^{\gamma(1/2-\si)} ,
\end{align*}
since $0<\si<1/2$. 
Optimizing the estimate in $\si$, we find
\begin{equation*}
  \si = \g\(\frac{1}{2}-\si\)\Longleftrightarrow \si = \frac{\g}{2(1+\g)},
\end{equation*}
which is consistent with $0<\si<1/2$. 
Then Gagliardo--Nirenberg inequality yields
\begin{equation*}
\|w^{\eps}(t)\|_{L^{r}(\R^{d})}\lesssim \eps^{-\gamma/4}
\|w^{\eps}(t)\|_{L^{2}(\R^{d})}^{1-\gamma/4}\|\eps
\nabla w^{\eps}(t)\|_{L^{2}(\R^{d})}^{\gamma/4}\lesssim
\eps^{\frac{\g}{2(1+\g)}-\gamma/4}.
\end{equation*}
Now we notice that
\begin{equation*}
  \frac{\g}{2(1+\g)}-\frac{\gamma}{4}>-\frac{\g}{8}, 
\end{equation*}
so for any $T>0$ independent of $\eps$, we can find $\eps_0$ so that 
\eqref{7.4} holds for $t\in [0,T]$ provided that $\eps\in ]0,\eps_0]$.
Theorem ~\ref{thm:superpositionbounded} follows.

\subsection{Proof of Theorem ~\ref{thm:superpositionlargetime}}
Since $a_{j}\in\Sigma^{k}$, we have, by Proposition~\ref{prop:global},
\begin{equation*}
M_{k}(t)\le C_{k}e^{C_{k}t}.
\end{equation*}
From Lemma 6.3 in \cite{CaFe11}, there exists positive
constants $C, C_{1}, C_{2}$ independent of $\eps$ such that
\begin{equation*}
|I_{\eps}(t)|\le C\eps^{\si}e^{C_{1}t}|E_{1}-E_{2}|^{-2}, \quad
\forall t\in \left[0, C_{2}\ln\frac{1}{\eps}\right].
\end{equation*}
It follows from Lemma \ref{lem:7.1} that
\begin{equation*}
\eps^{-1}\|N_{I}^{\eps}\|_{L^{1}([0, t]; \Sigma_{\eps})} \lesssim
e^{Ct}(t\eps^{\gamma(1/2-\si)}+
\eps^{\si}e^{C_{1}t})\lesssim
e^{Ct}(\eps^{\gamma(1/2-\si)}+
\eps^{\si})
\end{equation*}
Choosing $k\ge  3$ and the same (optimal) $\si$ as before,
\begin{equation*}
\si=\frac{\gamma}{2(1+\gamma)},
\end{equation*}
we have
\begin{equation*}
\eps^{-1}\|N_{I}^{\eps}\|_{L^{1}([0, t]; \Sigma_{\eps})}
\lesssim\eps^{\gamma/2(1+\gamma)}e^{Ct}.
\end{equation*}
 Resuming the bootstrap argument and similar arguments as in
Section ~\ref{sec:largetime} yields Theorem~\ref{thm:superpositionlargetime}.

 \bibliographystyle{amsplain}
\bibliography{biblio}

\providecommand{\bysame}{\leavevmode\hbox to3em{\hrulefill}\thinspace}
\providecommand{\MR}{\relax\ifhmode\unskip\space\fi MR }
\providecommand{\MRhref}[2]{%
  \href{http://www.ams.org/mathscinet-getitem?mr=#1}{#2}
}
\providecommand{\href}[2]{#2}
\begin{thebibliography}{10}

\bibitem{APPP-p}
A.~Athanassoulis, T.~Paul, F.~Pezzotti, and M.~Pulvirenti, \emph{Coherent
  states propagation for the {H}artree equation}, Ann. Henri Poincar\'e (2011),
  to appear. Archived at \url{http://arxiv.org/abs/1010.4889}.

\bibitem{BGP99}
D.~Bambusi, S.~Graffi, and T.~Paul, \emph{Long time semiclassical approximation
  of quantum flows: a proof of the {E}hrenfest time}, Asymptot. Anal.
  \textbf{21} (1999), no.~2, 149--160.

\bibitem{BR01}
J.~M. Bily and D.~Robert, \emph{The semi-classical {V}an {V}leck formula.
  {A}pplication to the {A}haronov-{B}ohm effect}, Long time behaviour of
  classical and quantum systems ({B}ologna, 1999), Ser. Concr. Appl. Math.,
  vol.~1, World Sci. Publ., River Edge, NJ, 2001, pp.~89--106.

\bibitem{CaBook}
R.~Carles, \emph{Semi-classical analysis for nonlinear {S}chr\"odinger
  equations}, World Scientific Publishing Co. Pte. Ltd., Hackensack, NJ, 2008.

\bibitem{Ca-p}
\bysame, \emph{Nonlinear {S}chr{\"o}dinger equation with time dependent
  potential}, preprint. Archived at \url{http://arxiv.org/abs/0910.4893}, 2009.

\bibitem{CaFe11}
R.~Carles and C.~Fermanian-Kammerer, \emph{Nonlinear {C}oherent {S}tates and
  {E}hrenfest {T}ime for {S}chr\"odinger {E}quations}, Commun. Math. Phys.
  \textbf{301} (2011), no.~2, 443--472.

\bibitem{CazCourant}
T.~Cazenave, \emph{Semilinear {S}chr\"odinger equations}, Courant Lecture Notes
  in Mathematics, vol.~10, New York University Courant Institute of
  Mathematical Sciences, New York, 2003.

\bibitem{CR97}
M.~Combescure and D.~Robert, \emph{Semiclassical spreading of quantum wave
  packets and applications near unstable fixed points of the classical flow},
  Asymptot. Anal. \textbf{14} (1997), no.~4, 377--404.

\bibitem{CR06}
\bysame, \emph{Quadratic quantum {H}amiltonians revisited}, Cubo \textbf{8}
  (2006), no.~1, 61--86.

\bibitem{CR07}
\bysame, \emph{A phase-space study of the quantum {L}oschmidt echo in the
  semiclassical limit}, Ann. Henri Poincar\'e \textbf{8} (2007), no.~1,
  91--108.

\bibitem{Fujiwara79}
D.~Fujiwara, \emph{A construction of the fundamental solution for the
  {S}chr\"odinger equation}, J. Analyse Math. \textbf{35} (1979), 41--96.

\bibitem{Fujiwara}
\bysame, \emph{Remarks on the convergence of the {F}eynman path integrals},
  Duke Math. J. \textbf{47} (1980), no.~3, 559--600.

\bibitem{PG96}
P.~G{\'e}rard, \emph{Oscillations and concentration effects in semilinear
  dispersive wave equations}, J. Funct. Anal. \textbf{141} (1996), no.~1,
  60--98.

\bibitem{GV85}
J.~Ginibre and G.~Velo, \emph{Scattering theory in the energy space for a class
  of nonlinear {S}chr\"odinger equations}, J. Math. Pures Appl. (9) \textbf{64}
  (1985), no.~4, 363--401.

\bibitem{H80}
G.~A. Hagedorn, \emph{Semiclassical quantum mechanics. {I}. {T}he {$\hbar
  \rightarrow 0$} limit for coherent states}, Comm. Math. Phys. \textbf{71}
  (1980), no.~1, 77--93.

\bibitem{HJ00}
G.~A. Hagedorn and A.~Joye, \emph{Exponentially accurate semiclassical
  dynamics: propagation, localization, {E}hrenfest times, scattering, and more
  general states}, Ann. Henri Poincar\'e \textbf{1} (2000), no.~5, 837--883.

\bibitem{HJ01}
\bysame, \emph{A time-dependent {B}orn-{O}ppenheimer approximation with
  exponentially small error estimates}, Comm. Math. Phys. \textbf{223} (2001),
  no.~3, 583--626.

\bibitem{Kato87}
T.~Kato, \emph{On nonlinear {S}chr\"odinger equations}, Ann. IHP (Phys.
  Th\'eor.) \textbf{46} (1987), no.~1, 113--129.

\bibitem{KT}
M.~Keel and T.~Tao, \emph{Endpoint {S}trichartz estimates}, Amer. J. Math.
  \textbf{120} (1998), no.~5, 955--980.

\bibitem{L86}
R.~G. Littlejohn, \emph{The semiclassical evolution of wave packets}, Phys.
  Rep. \textbf{138} (1986), no.~4-5, 193--291.

\bibitem{P97}
T.~Paul, \emph{Semi-classical methods with emphasis on coherent states},
  Quasiclassical methods ({M}inneapolis, {MN}, 1995), IMA Vol. Math. Appl.,
  vol.~95, Springer, New York, 1997, pp.~51--88.

\bibitem{Rob10}
D.~Robert, \emph{On the {H}erman-{K}luk semiclassical approximation}, Rev.
  Math. Phys. \textbf{22} (2010), no.~10, 1123--1145.

\bibitem{Rou-p}
V.~Rousse, \emph{Semiclassical simple initial value representations}, to appear
  in Ark. f\"or Mat. Archived as \url{http://arxiv.org/abs/0904.0387}, 2009.

\bibitem{SR09}
T.~Swart and V.~Rousse, \emph{A mathematical justification for the
  {H}erman-{K}luk propagator}, Comm. Math. Phys. \textbf{286} (2009), no.~2,
  725--750.

\bibitem{Yajima87}
K.~Yajima, \emph{Existence of solutions for {S}chr\"odinger evolution
  equations}, Comm. Math. Phys. \textbf{110} (1987), 415--426.

\end{thebibliography}

 \end{document}